\documentclass[11pt]{article}

\def\colorful{0}

\usepackage{times}
\usepackage{amsthm,amsfonts,amsmath,amssymb,epsfig,xcolor,graphicx,verbatim}
\usepackage{multirow}

\usepackage{enumitem}

\usepackage{framed}
\usepackage{nicefrac}

\usepackage[linesnumbered,ruled]{algorithm2e}

\def\nnewcolor{1}
\ifnum\nnewcolor=1

\fi
\ifnum\nnewcolor=0

\fi

\ifnum\colorful=1
\newcommand{\new}[1]{{\color{red} #1}}

\newcommand{\toremove}[1]{\new{\st{#1}}}
\else
\newcommand{\new}[1]{{#1}}

\newcommand{\toremove}[1]{}
\fi
\newcommand{\dnew}[1]{{ #1}}

\newif\ifhyper\IfFileExists{hyperref.sty}{\hypertrue}{\hyperfalse}
\hypertrue
\ifhyper\usepackage{hyperref}\fi
\usepackage{cleveref}

\newtheorem{theorem}{Theorem}[section]
\newtheorem{question}{Question}[section]

\newtheorem{lemma}[theorem]{Lemma}

\newtheorem{corollary}[theorem]{Corollary}
\newtheorem{claim}[theorem]{Claim}
\newtheorem{fact}[theorem]{Fact}
\crefname{fact}{fact}{facts}
\Crefname{fact}{Fact}{Facts}

\newtheorem{observation}[theorem]{Observation}
\theoremstyle{definition}
\newtheorem{definition}[theorem]{Definition}
\newtheorem{assumption}[theorem]{Assumption}

\newcommand{\R}{\mathbb{R}}
\newcommand{\Z}{\mathbb{Z}}
\newcommand{\E}{\mathbf{E}}
\newcommand{\var}{\mathbf{Var}}

\newcommand{\dtv}{d_{\mathrm TV}}

\newcommand{\ignore}[1]{}

\newcommand{\eps}{\epsilon}

\newcommand{\pr}{\mathbf{Pr}}

\newcommand{\Bin}{\mathop{\mathrm{Binom}}\nolimits}
\newcommand{\Poi}{\mathop{\mathrm{Poi}}\nolimits}
\newcommand{\Mult}{\mathop{\mathrm{Multinom}}\nolimits}

\renewcommand{\eqref}[1]{(\ref{#1})}

\newcommand{\eqdef}{\stackrel{{\mathrm {\footnotesize def}}}{=}}

\newcommand{\floor}[1]{\lfloor#1\rfloor}
\newcommand{\abs}[1]{\lvert#1\rvert}

\begin{document}

\title{Optimal Testing of Discrete Distributions with High Probability}

\author{
Ilias Diakonikolas\thanks{Supported by NSF Award CCF-1652862 (CAREER), 
NSF AiTF Award CCF-2006206, and a Sloan Research Fellowship.}\\
University of Wisconsin-Madison\\
{\tt ilias@cs.wisc.edu}\\
\and
Themis Gouleakis\thanks{Some of this work was performed while the author was a postdoctoral researcher at USC.}\\
MPI, Germany\\
{\tt tgouleak@mpi-inf.mpg.de}\\
\and 
Daniel M. Kane\thanks{Supported by NSF Award CCF-1553288 (CAREER) and a Sloan Research Fellowship.}\\
University of California, San Diego\\
{\tt dakane@cs.ucsd.edu}\\
\and
John Peebles\thanks{Supported by the Swiss National Science Foundation Grant \#200021\_182527.}\\ 
Yale University\\
{\tt john.peebles@yale.edu}\\
\and
Eric Price\thanks{Supported by NSF Award CCF-1751040 (CAREER).}\\
UT Austin\\
{\tt ecprice@cs.utexas.edu }\\
}

\maketitle

\setcounter{page}{0}
\thispagestyle{empty}

\begin{abstract}
We study the problem of testing discrete distributions with a focus on the high probability regime.
Specifically, given samples from one or more discrete distributions, a property $\mathcal{P}$, and 
parameters $0< \eps, \delta <1$, we want to distinguish {\em with probability at least $1-\delta$}
whether these distributions satisfy $\mathcal{P}$ or are $\eps$-far from $\mathcal{P}$
in total variation distance. Most prior work in distribution testing studied the constant confidence case 
(corresponding to $\delta = \Omega(1)$), and provided sample-optimal testers for a range of properties.
While one can always boost the confidence probability of any such tester by black-box amplification, 
this generic boosting method typically leads to sub-optimal sample bounds.

Here we study the following broad question: For a given property $\mathcal{P}$, can we {\em characterize} 
the sample complexity of testing $\mathcal{P}$ as a function of all relevant problem parameters, 
including the error probability $\delta$? Prior to this work, uniformity testing was the only statistical task
whose sample complexity had been characterized in this setting. As our main results,
we provide the first algorithms for closeness and independence testing that are sample-optimal, within 
constant factors, as a function of all relevant parameters. We also show matching
information-theoretic lower bounds on the sample complexity of these problems.
Our techniques naturally extend to give optimal testers for  related problems. To illustrate the generality of our methods, 
we give optimal algorithms for testing collections of distributions and testing closeness with unequal sized samples.
\end{abstract}

%%%%%%%%%%%%%%%%%%%%%%%%%%%%%%%%%%%%%%%%%%%%%%%%%%%%%%%%%%%%%%%%%%%%%%%%%%%%%%%%%%%%%%%%%%%%%%%%%%

\newpage

\section{Introduction} \label{sec:intro}

\subsection{Background and Motivation} \label{ssec:background}
This paper studies problems in distribution property testing~\cite{GR00, BFR+:00, Batu13},
a field at the intersection of property testing~\cite{RS96, GGR98} and
statistical hypothesis testing~\cite{NeymanP, lehmann2005testing}.
The prototypical problem of this field is the following: Given sample access to a collection of unknown
probability distributions and a pre-specified global property $\mathcal{P}$ of these distributions,
determine whether the distributions satisfy $\mathcal{P}$ or are ``far''
from satisfying the property. (See Section \ref{ssec:results} for a formal definition.) The main goal is to characterize
the sample and computational complexity of this general question, 
for any given property $\mathcal{P}$ of interest,
as a function of the relevant parameters. 
During the past two decades, distribution property testing has received significant attention
within the computer science and statistics communities. The reader is referred
to~\cite{Rub12, Canonne15} for two surveys on the topic.
It should be noted that the TCS definition of distribution testing 
is equivalent to the minimax view of statistical hypothesis testing, 
pioneered in the statistics community
by Ingster and coauthors (see, e.g.,~\cite{IS02}.)

The vast majority of prior research in distribution testing
focused on characterizing the complexity of testing various properties
of arbitrary discrete distributions
in the ``constant confidence regime.'' That is,
the testing algorithm is allowed to fail with probability (say) at most $1/3$.
This regime is by now fairly well understood:
For a range of natural and important
properties (see, e.g.,~\cite{Paninski:08, CDVV14, VV14, DKN:15, DKN:15:FOCS,
ADK15, DK16, DiakonikolasGPP16, CDKS18, Neyk20}),
prior work has developed testers with provably optimal sample complexity (up to universal constant factors).
More recently, a body of work has focused on leveraging a priori structure
of the underlying distributions to obtain significantly improved sample
complexities~\cite{BKR:04, DDSVV13, DKN:15, DKN:15:FOCS, CDKS17, DaskalakisP17, DaskalakisDK16, DKN17, DiakonikolasKP19}.
Similarly, all these results on testing structured distributions study the constant confidence regime.

%\inote{First attempt at motivating the problem.}

Since distribution property testing is a (promise) decision problem, one can use standard amplification
to boost the confidence probability of any tester to any desired value in a black-box manner.
Suppose we have a testing algorithm for property $\mathcal{P}$ that guarantees
confidence probability $2/3$ (failure probability $1/3$) with $N$ samples.
Using amplification, we can increase the confidence
probability to $1-\delta$, for any $\delta>0$, by increasing the sample complexity of the algorithm
by a factor of $\Theta(\log(1/\delta))$. In part due to this simple fact, the initial definition of property testing~\cite{GGR98}
had set the confidence parameter $\delta$ to be constant by default. As Goldreich
notes~\cite{Gold17-whp}, ``eliminating the error probability as a
parameter does not allow to ask whether or not one may improve over
the straightforward error reduction''. Indeed, as we will see below, for a range of tasks
this $\Theta (\log(1/\delta))$ multiplicative increase in the sample size is sub-optimal.

The previous paragraph leads us to the following general question:

\begin{question}\label{qn:main}
For a given property $\mathcal{P}$, can we {\em characterize} the sample complexity of testing $\mathcal{P}$
as a function of all relevant problem parameters, including the error probability $\delta$?
\end{question}

We believe that Question~\ref{qn:main} is of fundamental theoretical and practical interest
that merits investigation in its own right.  The analogous question
in the context of {\em distribution learning} has been intensely studied in statistical learning theory
(see, e.g.,~\cite{VWellner96, DL:01}) and tight bounds are known in a range of settings.

Question~\ref{qn:main} is of substantial interest
in statistical hypothesis testing, where the
family of distribution testing algorithms with failure probability
$\delta$ for a given property $\mathcal{P}$ is equivalent to the family of
minimax statistical tests whose probability of Type I error ($p$-value) and
probability of Type II error are both at most $\delta$.
Standard techniques for addressing
the problem of multiple comparisons, such as Bonferroni correction,
require vanishingly small $p$-values. In such settings, obtaining optimal testers in the high-confidence
regime might have practical implications in application areas
of hypothesis testing (e.g., in biology).

It should be noted that Question~\ref{qn:main} has received renewed research attention in the information theory
and statistics communities. Specifically,~\cite{HM13-it, Kim20}
focused on developing testers with improved dependence on $\delta$ for uniformity testing~\cite{HM13-it},
equivalence and independence testing~\cite{Kim20}. Prior to this work, uniformity testing---and, via Goldreich's reduction~\cite{Goldreich16}, identity testing---was the only statistical task
whose sample complexity had been characterized in the high-confidence regime~\cite{DGPP17}.
As shown in~\cite{DGPP17}, all previously studied uniformity testers are in fact sub-optimal
in the high-confidence regime. In other words, obtaining an optimal sample bound 
was not just a matter of improved analysis, but a new algorithm was required.

Most relevant to the results of this paper is the concurrent work by Kim, Balakrishnan, and Wasserman~\cite{Kim20}.
Kim {\em et al.}~\cite{Kim20} give equivalence and independence testers for discrete distributions
with respect to the total variation distance (i.e., in the same setting as ours)
whose sample complexities beat standard amplification as a function of $\delta$ (in some parameter regimes).
As we show in this paper, their sample complexity upper bounds are sub-optimal -- by roughly a quadratic factor.
See Section~\ref{ssec:related} for a detailed description of the most relevant prior work.

\subsection{Our Contributions} \label{ssec:results}

In this work, we systematically investigate the sample complexity of distribution testing
in the high-confidence regime. Our main focus is on the problems of closeness (equivalence) testing
and independence testing. We develop new techniques that lead to
the first sample-optimal testing algorithms for these properties. Moreover, we prove 
information-theoretic lower bounds showing that the sample complexity of our algorithms
is optimal in all parameters (within a constant factor). Our techniques can be naturally adapted
to give sample-optimal testers for other properties. To illustrate the generality of our methods, 
we show that our techniques lead to sample-optimal testers (and matching lower bounds) 
for testing properties of collections of distributions and testing closeness with unequal sized samples.

We start with a general definition of distribution property testing for tuples of distributions.

\begin{definition}[$(\eps, \delta)$-testing of property $\mathcal{P}$] \label{def:testing}
Let $\mathcal{P}$ be a property of a $k$-tuple of distributions.
Given parameters $0< \eps, \delta < 1$, and sample access to a collection of distributions $p^{(1)}, \ldots, p^{(k)}$,
we want to distinguish {\em with probability at least $1-\delta$} between the following cases:
\begin{itemize}
\item {\bf Completeness}: $(p^{(1)}, \ldots, p^{(k)}) \in \mathcal{P}$.
\item {\bf Soundness}: $(p^{(1)}, \ldots, p^{(k)})$ is $\eps$-far from $\mathcal{P}$, in total variation distance, 
i.e., for every $(q^{(1)}, \ldots, q^{(k)}) \in \mathcal{P}$ the average total variation distance between $p^{(i)}$
and $q^{(i)}$, $i \in [k]$, is at least $\eps$.
\end{itemize}
We call this the problem of {\em $(\eps, \delta)$-testing property $\mathcal{P}$}.
An algorithm that solves this problem will be called an \emph{$(\varepsilon,\delta)$-tester} for property $\mathcal{P}$.
\end{definition}

Here we focus on testing properties of distributions on discrete domains.
Definition~\ref{def:testing} captures all testing tasks we study in this paper.
Our contributions are described in detail in the proceeding discussion.

\medskip

The task of closeness testing (or equivalence testing) of two discrete distributions $p, q$ supported on $[n]$
corresponds to the case $k=2$ of Definition~\ref{def:testing} and the property in question is 
$\mathcal{P} = \{ (p, q): p = q \}$. In other words, given samples from $p$ and $q$,
we want to distinguish between the cases that $p = q$ and $\dtv(p, q) \geq \eps$.
For closeness testing, we show:

\begin{theorem}[Closeness Testing] \label{thm:close-main}
There exists a computationally efficient $(\eps, \delta)$-closeness tester for discrete distributions of support size $n$
with sample complexity
$$\Theta\left( n^{2/3}\log^{1/3}(1/\delta) / \eps^{4/3} + (n^{1/2} \log^{1/2}(1/\delta)+\log(1/\delta))/\eps^2 \right) \;.$$
Moreover, this sample size upper bound is information-theoretically optimal, 
within a universal constant factor, for all $n, \eps, \delta$.
\end{theorem}

\medskip

The statistical task of (two-dimensional) independence testing of a discrete distribution $p$ on the domain $[n] \times [m]$
corresponds to the case $k=1$ of Definition~\ref{def:testing}, where the property of interest is 
 $\mathcal{P} = \{p: p \textrm{ is a product distribution}\}$. That is, we want to distinguish between the case that $p$
 is a product distribution versus $\eps$-far, in total variation distance, 
 from any product distribution. For independence testing, we show:

\begin{theorem}[Independence Testing] \label{thm:indep-main}
There exists a computationally efficient $(\eps, \delta)$-independence tester for discrete distributions on $[n] \times [m]$,
where $n \geq m$, with sample complexity
$$\Theta\left( n^{2/3} m^{1/3} \log^{1/3}(1/\delta) / \eps^{4/3} + ((nm)^{1/2}\log^{1/2}(1/\delta)+\log(1/\delta))/\eps^2 \right) \;.$$
Moreover, this sample size upper bound is information-theoretically optimal, within a universal constant factor, for all $n, m, \eps, \delta$.
\end{theorem}

The main focus of this paper is on developing the techniques required to establish 
Theorems~\ref{thm:close-main} and~\ref{thm:indep-main}. Building on these techniques,
we obtain optimal testers for two additional fundamental properties. 

In the task of testing collections of distributions, we are given access to $m$ distributions 
$p^{(1)}, \ldots, p^{(m)}$ supported on $[n]$ and we want to distinguish between the case that 
$p^{(1)} =  p^{(2)}  = \ldots = p^{(m)}$ and the case that $\min_{q} (1/m) \sum_{i=1}^m \dtv(p^{(i)} , q) \geq \eps$.
Our algorithm is given samples of the form $(i, j)$, where $i$ is drawn uniformly at random from $[m]$ 
and $j \in [n]$ is drawn from $p^{(i)}$. While this problem has strong similarities to independence testing,
it also has some significant differences. For this testing task, we show:

\begin{theorem}[Testing Collections of Distributions] \label{thm:coll-main}
There exists a computationally efficient $(\eps, \delta)$-tester for testing closeness of collections of $m$ distributions on $[n]$ 
with sample complexity
$$\Theta\left( n^{2/3} m^{1/3} \log^{1/3}(1/\delta) / \eps^{4/3} + ((nm)^{1/2}\log^{1/2}(1/\delta)+\log(1/\delta))/\eps^2 \right) \;.$$
Moreover, this sample size upper bound is information-theoretically optimal, within a universal constant factor, for all $n, m, \eps, \delta$.
\end{theorem}

Our final result is for the problem of testing closeness between two unknown discrete distributions when we have access
to unequal sized samples from the two unknown distributions. This problem interpolates between the vanilla closeness testing task
(with equal sized samples) and the task of identity testing (where one of the two distributions is known exactly). 
For this task, we show:

\begin{theorem}[Closeness Testing with Unequal Sized Samples] \label{thm:close-unequal-main}
There exists a computationally efficient $(\eps, \delta)$-closeness tester for discrete distributions of support size $n$
that draws $O(K+k)$ samples from one distribution and $O(k)$ samples from the other, as long as 
$$k \geq C \left(n \sqrt{\log(1/\delta) / \min(n, K)} + \log(1/\delta)\right)/\eps^2 \;,$$
where $C>0$ is a universal constant.
Moreover, this sample size tradeoff is information-theoretically optimal, 
within a universal constant factor, for all $n, \eps, \delta$.
\end{theorem}

\subsection{Overview of Techniques} \label{ssec:techniques}

In this section, we provide a detailed overview of our upper and lower bound
techniques. Our main technical and conceptual innovation lies in the development of our 
upper bounds. To keep this section concrete, we describe our techniques in the context
of closeness and independence testing. Our algorithms for testing collections and closeness
with unequal sized samples use very similar ideas to those of our independence tester. 

\paragraph{Closeness Tester.}
To obtain a closeness tester that performs well
in the high confidence regime, we need to design a  
test statistic that exhibits strong concentration bounds. 
A reasonable approach to enforce this requirement would be to 
ensure that the test statistic is Lipschitz in the samples, 
so that we can leverage an appropriate concentration inequality
(e.g., McDiarmid's inequality) to obtain the necessary concentration. 
We note that the chi-squared closeness tester of~\cite{CDVV14} is Lipschitz,
but not Lipschitz enough for the straightforward analysis to obtain an optimal bound.
While we conjecture that the
\cite{CDVV14} closeness tester is indeed optimal, here we develop
a new and easier to analyze closeness tester. Our new closeness
tester (and its analysis) will also be crucially used for our independence tester. 

We are now ready to describe the new statistic that 
our closeness tester relies on.
Let $X_i, Y_i$ be the number of samples assigned to bin (domain element) $i \in [n]$,
from $p$ and $q$ respectively. A natural starting point is to consider the absolute value of the
difference $|X_i - Y_i|$. Namely, we could consider the statistic $Z = \sum_{i=1}^n |X_i-Y_i|$ and 
output ``YES'' or ``NO'' based on its magnitude. 
Unfortunately, this random variable $Z$ does not have mean zero in the completeness case
(i.e., when $p=q$). Furthermore, one can construct instances where the expectation of this 
statistic is not even minimized when $p=q$.
To fix this issue, we will need to subtract an appropriate proxy for what the value
should be if $p=q$. To do this, we draw a second set of samples with $X'_i$
and $Y'_i$ samples in bin $i$ from each of the distributions. We then use
the test statistic
$$Z = \sum_{i=1}^n \left( |X_i-Y_i| + |X'_i-Y'_i| - |X_i - X'_i| - |Y_i- Y'_i| \right) \;.$$
If $p=q$, it is clear that $X_i,X'_i,Y_i,Y'_i$ are i.i.d., and so $Z$ is
mean zero. The challenging part of the proof involves showing that if $p$ is $\eps$-far from $q$, then 
$\E[Z]$ must be large. Since $Z$ is Lipschitz, it satisfies strong concentration bounds, 
and so with sufficiently many samples we can distinguish the two cases with high probability.
A careful analysis shows that this tester is indeed sample optimal for the entire parameter regime.

\paragraph{Independence Tester.}
Let $p$ be a discrete distribution on $[n] \times [m]$.
It is easy to see (and well-known) that
the independence testing problem amounts to distinguishing the case
where $p=q$ from the case that $p$ is $\eps$-far from $q$, where $q$
is the product of $p$'s marginals. Unfortunately, directly applying
Theorem~\ref{thm:close-main} to this domain of size $nm$ gives a poor
sample complexity in one of the three terms. In particular, the first term 
would be $n^{2/3} m^{2/3}$, not $n^{2/3} m^{1/3}$.  
Of course, this is an issue even for the constant confidence regime.
We thus need a better bound when this term is dominant, 
which we will obtain using tighter concentration bounds on our statistic 
$Z$ from the previous subsection.

We start by observing that if $Z$ is computed by drawing a total of
$k$ independent samples, the fact that $Z$ is Lipschitz implies a
variance bound of $O(k)$.  By McDiarmid's inequality, it follows that
$Z$ is within $O(\sqrt{k\log(1/\delta)})$ of its mean value with
probability $1-\delta$.  However, we note that the value of the output
statistic for $Z$ does not really depend on all of the samples.  In
particular, any bin (domain element) with exactly one sample drawn
from it (from the combination of $p$ and $q$) will not contribute to
the statistic. Hence, if we let $N$ be the number of non-isolated
samples, then in some sense, the variance of $Z$ will be bounded above
by $N$. Formally speaking, some technical work is needed here, because we may have gotten
unlucky and drawn samples with an unusually small value of $N$. To address this, 
we use a symmetrization argument to show that
$|Z-\E[Z]|=O(\sqrt{(N+\log(1/\delta))\log(1/\delta)})$ with
probability at least $1-\delta$ (see Lemma~\ref{lem:Zconc}). If we can
ensure that the number of non-isolated samples is not too large, this
stronger concentration bound should allow us to use fewer samples.

In order to decrease the number of non-singleton samples in our distribution, 
it is natural to want our underlying distributions 
to have small $\ell_2$ norm. An approach to achieve this is by 
using the flattening technique of~\cite{DK16}. The basic idea of flattening is to use some 
of our samples to identify the heavy bins in our distribution, and then to artificially 
subdivide these bins in order to decrease the total $\ell_2$ norm of the distribution. 
This technique is especially  useful for the product distribution $q$, as we can separately 
identify the heavy $x$-coordinates and heavy $y$-coordinates, 
rather than using what would need to be substantially more samples to identify all of the heavy pairs. 
{\em However, there are two major difficulties with using flattening in this setting. To circumvent these 
obstacles, new ideas are needed, as explained in the proceeding discussion.}

\begin{comment}
To make this
case work we will need a stronger version of [inequality name]. In
particular, we use Freedman's Inequality which says roughly that if $X$
is a martingale whose individual terms have contribution at most $1$ and
whose variance (measured in a path-dependent way) is at most sigma,
then $Pr(|X| > t) < exp(-Omega(min(t^2/sigma,t)))$ (this is similar to a
martingale version of [whatever the name is inequality]). It is not
hard to see that for our test statistic $Z$, this path-dependent
variance is roughly the number $N$ of non-isolated samples (i.e. samples
that collide with at least one other sample). This is essentially
because bins with at most $1$ sample do not contribute to $Z$. This
quantity $N$ unfortunately might have a large expectation if $p$ is far
from uniform. However, we can use the flattening technique of
[citation] to reduce it. There are two main issues with this.
\end{comment}

First, although flattening can be used to reduce the number of
collisions coming from samples of $q$, it will not necessarily reduce
the number of collisions from $p$-samples to acceptable levels. We get
around this issue by noting that if most of the collisions contributing to $N$
come largely from $p$-samples, then with high probability it will be case
that $Z \gg N$, in which case the larger variance term will not hurt us much.
A second, more difficult, problem to handle is this: although it is not hard to
show that flattening works {\em on average}, it simply is not true
that flattening yields a small number of collisions with sufficiently
high probability. This is a major issue in our setting, since our goal 
is to obtain the optimal sample complexity with high confidence!

To circumvent the latter problem, we will need to
substantially restructure our algorithm.  Essentially, we will pick a
set $S$ of samples once at the beginning of our algorithm. We then
randomly assign samples of $S$ to be used either to flatten $x$ and
$y$ coordinates, or to generate samples from $p$ and $q$.  If we got
unlucky and our flattening was not sufficient (because the number of
$q$-samples that collided was too large), we will try again using the
same initial set $S$ of samples, but re-randomizing the way these
samples are used.

To show that this new algorithm works, we will need to establish two statements:
\begin{enumerate}
\item For \emph{any} set of initial samples $S$, the probability that we will need to try again 
is at most $50\%$ (so, on average, we only need to try a constant number of times).
\item The probability that a given try causes our algorithm to
terminate with the wrong answer is at most $\delta$.
\end{enumerate}
Combining the second statement with the fact that on average we will
only need $O(1)$ many tries before we get an answer, the total
probability of failure will be bounded by
$\delta \E[\text{\# tries}] = O(\delta).$ This allows us to get a
high-probability bound even though our analysis of flattening only
works on average.

\paragraph{Sample Complexity Lower Bounds.}
We sketch our sample complexity lower bound for independence testing.
The corresponding lower bound for closeness testing follows as a special
case in a black-box manner.

Our lower bound proof follows the same outline as the lower bound proof in~\cite{DK16}. 
The gist of the argument in that work was that we reduced
to the following problem: We have two explicit pseudo-distributions\footnote{A ``pseudo-distribution'' is like a distribution, except not necessarily normalized to sum to one.}
$D_{\textrm{yes}}$ (over independent pseudo-distributions) and $D_{\textrm{no}}$ 
(over usually far from independent pseudo-distributions). We pick a random
pseudo-distribution from one of these families, take $\Poi(k)$ samples
from it, hand them to the algorithm, and ask the algorithm to determine
which ensemble we started with. It was shown in~\cite{DK16} that it is impossible to
do this reliably by bounding the {\em mutual information} between the samples and
the bit determining which ensemble was sampled from.

This approach, unfortunately, does not suffice for high probability bounds.
\new{\cite{DK16} worked in the constant confidence 
regime, where the mutual information is close to $0$.} In contrast, 
in the high confidence regime, the mutual information will be close to $1$.
While, in principle, bounding the mutual information away from $1$ might
suffice to prove lower bounds in the high confidence regime, the mutual
information bounds achievable with the~\cite{DK16} techniques 
are not sufficiently strong, in the sense that they can only bound
the mutual information by a quantity bigger than $1$, given enough samples. 

To overcome this technical hurdle, we replace our bounds on
mutual information with bounds on KL-divergence. Unlike the mutual
information (which is bounded by $1$ bit), the KL-divergence between our
distributions can become arbitrarily large. It is also not hard to see
that if two distributions can be distinguished with probability
$1-\delta$, the KL-divergence is $\Omega(\log(1/\delta))$. (See \Cref{lem:better-dtv-ub}.)

Given the above observation, our lower bound ensembles are identical
to the ones used in~\cite{DK16}. Furthermore, the analytic techniques
we use to bound the KL-divergence are very similar, using essentially
the same expression as an upper bound on KL-divergence as was used as
an upper bound on mutual information.  Another technical issue is that
we need to show that the reduction to our hard instance over
pseudo-distributions still works for high probability testing, which
is not difficult, but needs to be carefully checked.

\subsection{Prior and Concurrent Work} \label{ssec:related}
Prior to this work, the question of developing sample-optimal testers in the high-confidence regime has been considered
for uniformity testing (and, via Goldreich's reduction, identity testing). Specifically,~\cite{HM13-it} showed
that Paninski's uniformity tester (based on the number of unique elements) has the sample-optimal sample complexity
of $O(\sqrt{n \log(1/\delta)}/\eps^2)$ in the sublinear sample regime, i.e., when the sample size is $o(n)$. More recently,
~\cite{DGPP17} gave a different tester that achieves the optimal sample complexity $O((\sqrt{n \log(1/\delta)}+\log(1/\delta))/\eps^2)$
in the entire regime of parameters.

As already mentioned, prior to our work, uniformity was the only property for which the high confidence 
regime has been analyzed. We now comment on some closely related literature.~\cite{CDVV14} 
gave a chi-squared tester and showed that it is sample-optimal in the constant confidence regime. 
We believe that the same tester is optimal in the high-confidence regime.
However, a proof of this statement seems rather non-trivial. 
In particular, simple analyses based on McDiarmid's inequality \cite{mcdiarmid_1989} 
lead to sub-optimal sample complexity when the sample size is $\Omega(n)$.
The new closeness tester introduced in this work is arguably simpler 
with a compact analysis, and it is crucial for our much more involved independence tester.

The work of \cite{ADK15} gave an independence tester that is sample optimal '
in the constant confidence regime for the special 
case that the two dimensions have the same support size (i.e., $n=m$). 
The performance of this tester is sub-optimal in the high-confidence regime, 
as it relies on a non-Lipschitz identity tester. \cite{DK16} gave a sample-optimal independence tester 
for the general case (where $n \geq m$), which is the only known sample-optimal tester 
in the constant confidence regime for this problem. Unfortunately, this tester is also sub-optimal 
in the high-confidence regime for the following reason.
\cite{DK16} uses the flattening technique to reduce the problem under total variation ($\ell_1$) distance
to an $\ell_2$-closeness testing problem. The issue is that the $\ell_2$-testing task does not behave
well in the high probability regime, so this approach does not suffice to give optimal testers in this setting.  
While our optimal independence tester in this paper also leverages the flattening technique, it requires
several new conceptual and technical ideas.

Concurrent and independent work~\cite{Kim20} provided testers for
closeness and independence testing in the high-confidence
regime. Their algorithms distinguish between the Type 1 and Type 2
error probabilities $\alpha$ and $\beta$ respectively. Our results in
this paper correspond to the setting that $\alpha = \beta = \delta$.
Their testers have polynomial dependence on $1/\beta$ and therefore do
not perform well in our setting.  For constant $\beta$, their testers
perform better than naive amplification but still sub-optimally in
the parameter $\alpha$.  For example, their Theorem~8.1 gives a closeness tester 
with sample complexity of
$m = O(n^{2/3} \log^{2/3}(1/\alpha)/\beta^{4/3} + n^{1/2} \log(1/\alpha)/\beta^2)$.
Even for $\beta = \Theta(1)$, this is essentially quadratically worse
in $\log (1/\alpha)$ than applying Theorem~\ref{thm:close-main} with
$\delta = \alpha$.

\subsection{Organization} \label{sec:organ}
After setting up the required preliminaries in Section~\ref{sec:prelims},
we give our testing algorithms for closeness and independence
in Sections~\ref{sec:closeness-alg} and~\ref{sec:indep-alg}. In Section~\ref{sec:lb-ind}, we establish our sample complexity
lower bounds for these two problems. Our upper bounds for testing collections and closeness with unequal sized
samples are given in Appendices~\ref{sec:col} and~\ref{sec:close-unequal-alg} respectively. 
Finally, Appendix~\ref{sec:lb-unequal} proves our sample lower bound for closeness with unequal sized samples.

\section{Preliminaries} \label{sec:prelims}

\subsection{Notation}

We write $[n]$ to denote the set $\{1, \ldots, n\}$. 
We consider discrete distributions over $[n]$ with corresponding
probability mass functions $p: [n] \rightarrow [0,1]$ satisfying $\sum_{i=1}^n p_i =1.$ 
We use the notation $p_i$ to denote the probability of element
$i$ in distribution $p$. The $\ell_1$ (resp. $\ell_2$) norm of a distribution is identified with the $\ell_1$ (resp. $\ell_2$) 
norm of the corresponding vector, i.e., $\|p\|_1 = \sum_{i=1}^n p_i$ and $\|p\|_2 = \sqrt{\sum_{i=1}^n p^2_i}$. 
Similarly, the $\ell_1$ (resp. $\ell_2$) distance between distributions $p$ and $q$ is the $\ell_1$ (resp. $\ell_2$) 
norm of the vector of their difference. 
The total variation distance between distributions $p, q$ on $[n]$ is 
$\dtv(p, q) \eqdef \frac{1}{2} \cdot \|p-q\|_1$. The KL divergence between 
two discrete distributions $p$ and $q$ on $[n]$ is $D(p||q) =  \sum_i p_i \log(p_i/q_i)$.

A Poisson distribution with parameter $\lambda$ is denoted $\Poi(\lambda)$. The binomial and multinomial distributions are denoted $\Bin(n, p)$ and $\Mult(n, \{p_i\}_{i=1}^k)$, respectively.

%The \emph{histogram} of a set of samples is the vector $v$ where $v_i$ is the number of elements that were sampled $i$ times.

The main concentration inequality used in our upper bounds is McDiarmid's inequality.
\begin{fact}[McDiarmid's Inequality\cite{mcdiarmid_1989}]
Let $f$ be a multivariate function with $m$ independent random inputs whose codomain is
$\R$ and such that, for each $i \in [m]$, changing the $i$th
coordinate alone can change the output by at most $c_i$ additively. Then
$\pr[\left|f(X) - \E[f(x)]\right| \geq t] \leq 2e^{-\frac{2t^2}{\sum_i c_i^2}}.$
\end{fact}

A commonly used method for bounding from above the total variation distance in terms of KL divergence is Pinsker's inequality. 
However, Pinsker's inequality is mainly useful when the KL divergence is small. In the high probability regime, 
the KL divergence is larger than $1$ and this gives no information about the total variation distance. 
Our sample complexity lower bounds instead use a different inequality, which is better suited for the high probability regime.

\begin{fact}[see, e.g., Lemmas 2.1 and 2.6 of \cite{tsybakov_introduction_2009}] \label{lem:better-dtv-ub}
For any pair of distributions $p, q$, we have that
$\dtv(p,q) \leq 1 - (1/2) e^{-D(p||q)}$. Equivalently, it holds $D(p||q) \geq \log(2/\delta)$, 
where $1-\delta$ is the total variation distance.
\end{fact}

\section{Sample-Optimal Closeness Tester} \label{sec:closeness-alg}

In this section, we give our optimal closeness tester, described in pseudo-code below.
 
\medskip

\new{
\begin{algorithm}[H]%\label{alg:Basic-indep}
    \SetKwInOut{Input}{Input}
    \SetKwInOut{Output}{Output}
    \SetKwRepeat{Do}{do}{while}
\new{
    \Input{ sample access to distributions $p, q$ over $[n]$, $\eps>0$, and $\delta>0$.}
    \Output{``YES'' if $p =q$, ``NO'' if $\dtv(p, q) \geq \eps$; \new{both with probability at least $1-\delta$}.}
    \label{k:close-def} Set $k = C \left( n^{2/3}\log^{1/3}(1/\delta)/\eps^{4/3}+ \big(n^{1/2} \log^{1/2}(1/\delta) + \log(1/\delta)\big)/\eps^2 \right)$,
    where $C>0$ is a sufficiently large universal constant.\\
   Set $(\widetilde{m_p}, \widetilde{m_p}', \widetilde{m_q}, \widetilde{m_q}') = \Mult \left(4k, (1/4, 1/4, 1/4, 1/4)\right)$.\\
     Draw two multi-sets of independent samples from $p$ of sizes $\widetilde{m_p}, \widetilde{m_p}'$ respectively, 
          and two multi-sets of independent samples from $q$ of sizes $\widetilde{m_q}, \widetilde{m_q}'$ respectively.
          Let $\widetilde{X}=(\widetilde{X}_i)_{i=1}^n$,  $\widetilde{X}'=(\widetilde{X}'_i)_{i=1}^n$, 
          $\widetilde{Y}=(\widetilde{Y}_i)_{i=1}^n$, $\widetilde{Y}'=(\widetilde{Y}'_i)_{i=1}^n$
          be the corresponding histograms of the samples.  \\
          \label{Z:tilde-close-def} Compute the value of the random variable $\widetilde{Z}=\sum_{i=1}^n \widetilde{Z}_i$, 
            where, for $i \in [n]$, we define
            \begin{equation} \label{eqn:tilde-zi-closeness}
           \widetilde{Z}_i=|\widetilde{X}_i-\widetilde{Y}_i| + |\widetilde{X}'_i - \widetilde{Y}'_i | - |\widetilde{X}_i-\widetilde{X}'_i| 
           - |\widetilde{Y}_i-\widetilde{Y}'_i| \;.
           \end{equation}\\
           Set the threshold $T = C' \sqrt{k\log(1/\delta)}$, 
           where $C'$ is a universal constant (derived from the analysis of the algorithm).\\
           \eIf{$\widetilde{Z} \leq T$}{return ``YES"}{return ``NO"}

\caption{ \textsc{Test-Closeness$(p, q, n, \eps,\delta)$}%: Given a distribution $p$ over $[n]\times [m]$ (where $n\geq m$), % with marginals $p_x,p_y$, test if $p_x$ and $p_y$ are independent. 
%test if $p$ is a product distribution.
}
}
\end{algorithm}
}

\medskip

The main result of this section is the following theorem:

\begin{theorem}\label{thm:close-alg}
There exists a universal constant $C>0$ such that the following holds:
When 
\begin{equation} \label{eqn:k-close-bound}
k \geq C \left( n^{2/3}\log^{1/3}(1/\delta)/\eps^{4/3}+ \big(n^{1/2} \log^{1/2}(1/\delta) + \log(1/\delta)\big)/\eps^2 \right) \;,
\end{equation}
Algorithm \textsc{Test-Closeness} is an $(\eps, \delta)$-closeness tester in total variation distance.
\end{theorem}

To prove Theorem~\ref{thm:close-alg}, we will show that the
expected value of our statistic $\widetilde{Z}$ in the completeness case is sufficiently separated from the expected
value of $\widetilde{Z}$ in the soundness case, and also that the value of $\widetilde{Z}$ is highly concentrated around
its expectation in both cases. We proceed to prove these two steps in the following subsections.
We will assume that the parameter $k$ 
in Step~\ref{k:close-def} of the algorithm satisfies \eqref{eqn:k-close-bound}.

\subsection{Bounding the Expectation Gap} \label{ssec:close-exp-gap}

In this section, we will prove an $\Omega(\sqrt{k\log(1/\delta)})$ expectation 
gap between the completeness and soundness cases.
\new{We proceed by analyzing the expectation of a slightly modified random variable $Z$ obtained
by taking the number of samples drawn from $p$ and $q$ be Poisson distributed. We then relate
the expectation of $Z$ to the expectation of our actual statistic $\widetilde{Z}$.

\paragraph{Definition of modified random variable $Z$.}
Independently set $m_p= \Poi(k)$,  $m'_p= \Poi(k)$, $m_q= \Poi(k)$,  $m'_q= \Poi(k)$.
Draw two multi-sets of independent samples from $p$ of sizes $m_p, m'_p$ respectively, 
and two multi-sets of independent samples from $q$ of sizes $m_q, m'_q$ respectively.
Let $X=(X_i)_{i=1}^n$,  $X'=(X'_i)_{i=1}^n$, 
$Y=(Y_i)_{i=1}^n$, $Y'=(Y'_i)_{i=1}^n$
be the corresponding histograms of the samples.  
We will analyze the random variable %$Z=\sum_{i=1}^n Z_i$, where, for $i \in [n]$, we define
\begin{equation} \label{eqn:z-closeness}
Z=\sum_{i=1}^n Z_i, \textrm{ where } Z_i=|X_i-Y_i | + |X'_i - Y'_i | - |X_i-X'_i| - |Y_i-Y'_i| \;.
\end{equation}

Let $m = m_p+m'_p+m_q+m'_q$ be the total number of samples drawn from $p, q$ in the definition of 
$Z$. By construction, we have that $\widetilde{Z} = Z \mid (m=4k)$. This will allow
us to argue that $\E[Z]$ and $\E[\widetilde{Z}]$ are close to each other.

\begin{claim} \label{claim:exp-close-z-tilde}
We have that $|\E[Z] - \E[\widetilde{Z}]|  = O( \sqrt{k})$.
\end{claim}
\begin{proof}
%\inote{Maybe the stuff below needs a bit more explanation.}
Note that the statistic $Z$ is $2$-Lipschitz, i.e., adding a sample can change
$Z$ by at most $2$. Therefore, $|\E[Z \mid m=a] - \E[ Z \mid m=b]| \leq 2 |a-b|$. 
This implies that 
$$|\E[Z] - \E[\widetilde{Z}]|  = O( \E[ |m-4k|] ) = O( \sqrt{k})\;,$$
as desired.
\end{proof}

It therefore suffices to show that there is sufficient separation between $\E[Z]$ in the completeness
and soundness cases.} Specifically, this subsection is devoted to the proof of
the following lemma:

\begin{lemma}[Expectation Gap]\label{lem:close-exp-gap}
Let $Z$ be the statistic defined in~\eqref{eqn:z-closeness}. Then
\begin{itemize}
\item[(i)] If $p = q$ (completeness), we have that $\E[Z] = 0$.
\item[(ii)] If $\dtv(p, q) \geq \eps$ (soundness), we have that $\E[Z]=\Omega(\sqrt{k\log(1/\delta)})$.
%\inote{Do we want to write this a little differently? e.g., with an unspecified constant?}
\end{itemize}
\end{lemma}

Note that for each $i \in [n]$, $X_i, X'_i \sim \Poi(k p_i)$, $Y_i, Y'_i \sim \Poi(k q_i)$.
Moreover, the random variables $\{X_i, X'_i, Y_i, Y'_i\}_{i=1}^n$ are mutually independent.

The proof of Part (i) in Lemma~\ref{lem:close-exp-gap} is straightforward and holds  for all
$k \geq 1$. Since $p=q$, 
it follows that, for any fixed $i \in [n]$, the random variables $X_i, X'_i, Y_i, Y'_i$ are identically distributed.
Therefore, the random variables $|X_i-Y_i|$, $|X'_i-Y'_i|$, $|X_i-X'_i|$, and $|Y_i-Y'_i|$ are also
identically distributed, which implies that $\E[|X_i-Y_i|]=\E[|X'_i-Y'_i|]=\E[|X_i-X'_i|]=\E[|Y_i-Y'_i|]$.
Thus, $\E[Z_i] = 0$ for all $i \in [n]$, and therefore $\E[Z]=0$.

The proof of Part (ii) in Lemma~\ref{lem:close-exp-gap} is significantly more challenging.
We note that the proof of Part (ii) crucially relies on the assumption that $k$ is sufficiently large,
satisfying \eqref{eqn:k-close-bound}.

We start with the following technical claim:

\begin{claim}\label{cl:pointwise}
For all $i\in [n]$, we have that 
\begin{equation} \label{eqn:Zi-lb-close}
\E[Z_i]=\Omega\left( \min \left\{ |k p_i - k q_i|, |k p_i - k q_i|^2, \frac{|k p_i - k q_i|^2}{\sqrt{k p_i+k q_i}}\right\} \right) \;.
\end{equation}
\end{claim}
\begin{proof}
Recall that for each $i \in [n]$, $X_i, X'_i \sim \Poi(k p_i)$, $Y_i, Y'_i \sim \Poi(k q_i)$ and that
these random variables are mutually independent. This implies that $\E[|X_i - Y_i|] = \E[|X'_i - Y'_i|]$ and 
therefore
\[
\E[Z_i]= 2 \, \E[|X_i-Y_i|]-\E[|X_i-X'_i|]-\E[|Y_i-Y'_i|] \;.
\] 
\new{Due to the absolute values in the above expression,} we can assume without 
loss of generality that $a:= k p_i  \geq k q_i = : b$. 

Let $c: =a-b \geq 0$. Then we can write that 
$X_i, X'_i  \sim \Poi(b) + \Poi(c)$ and $Y_i, Y'_i  \sim \Poi(b)$. 
Let $B_1,B_2$ and $C_1,C_2$ be mutually independent random variables with 
$B_1, B_2 \sim \Poi(b)$ and  $C_1,C_2 \sim \Poi(c)$.  Note that $B_{\ell}+C_{\ell'}$, for $\ell, \ell' \in \{1,2\}$, 
have the same distribution as $X_i$ and $X'_i$. By linearity of expectation, we can thus write
\begin{eqnarray} 
\E[Z_i]=& (1/2) \; \E\Big[ |B_1+C_1-B_2|+|B_1+C_2-B_2|+|B_1-C_1-B_2|+|B_1-C_2-B_2| \nonumber \\ 
&- |B_1+C_1-B_2-C_2| - |B_1+C_2-B_2-C_1| - 2 |B_1-B_2| \Big]  \;, \label{eqn:Zi-symmetry}
\end{eqnarray}
where the first four terms above correspond to $2 \, \E[|X_i-Y_i|]$, the fifth
and sixth terms correspond to $-\E[|X_i-X'_i|]$, and the last term corresponds to $-\E[|Y_i-Y'_i|]$.

Consider the function $f:\R^2 \to \R$ defined as $f(x, y) = \new{(1/2)} \left( |x+y|+|y-x| -2 |y| \right)$. 
By the definition of $f$ and \eqref{eqn:Zi-symmetry}, we have that
\begin{equation} \label{eqn:zi-f}
\E[Z_i]=\E \left[ f(C_1, B_1-B_2)+f(C_2, B_1-B_2)-f(C_1-C_2, B_1-B_2) \right] \;.
\end{equation}
Now observe that $f(x, y)=  \max\{0, |x|-|y|\}$ and that $f(x, y)$ is an increasing function of $|x|$.

For any $x_1,x_2 \geq 0$ and $y \in \R$, we have that  $|x_1-x_2| \leq \max\{x_1, x_2 \}$, hence
$$f(x_1-x_2, y)  = f(|x_1-x_2|, y) \leq f(\max \{ x_1, x_2 \}, y)  =\max \{ f(x_1, y), f(x_2,y) \} \;.$$ 
This implies that 
\begin{eqnarray*}
f(x_1, y)+f(x_2, y) - f(x_1-x_2, y) 
&\geq& f(x_1, y)+f(x_2, y) - \max\{f(x_1, y), f(x_2, y)\} \\
&=& \min\{f(x_1, y), f(x_2, y)\} \\
&=& f \left(\min\{x_1,x_2\}, y \right)\;.
\end{eqnarray*}
Using \eqref{eqn:zi-f}, the above inequality gives that
\begin{equation}\label{eq:Zi-lb}
\E[Z_i] \geq  \E\left[f \left(\min\{C_1,C_2\}, B_1-B_2 \right) \right] 
= \E\left[\max\big\{ 0, \min \{C_1, C_2\} - |B_1-B_2| \big\} \right] \;.
\end{equation} 
Therefore, it suffices to establish a lower bound on the RHS of \eqref{eq:Zi-lb}.
We proceed to do so by considering two complementary cases, 
based on the value of the parameter $c\geq 0$.

\medskip

\noindent {\bf Case I: $c<1$.}

In this case, we can write
\begin{eqnarray*}
\E[Z_i] &\geq& \pr\left[(\min\{C_1,C_2\}\geq 1) \wedge (B_1=B_2)\right] = \pr[C_1\geq 1]^2 \pr[B_1=B_2]\\
&\geq& \Omega \left(c^2 \, \min \left\{1, 1/\sqrt{b} \right\} \right) = \Omega\left(\min\left\{c^2,c^2/\sqrt{b}\right\}\right) \;,
\end{eqnarray*}
where the first inequality follows from \eqref{eq:Zi-lb} 
(since $\min \{C_1, C_2\} - |B_1-B_2|\} \geq 1$ under the corresponding event), 
the first equality uses the independence of $B_1, B_2$ and $C_1$, and the 
last inequality uses the fact that $\pr[C_1 \geq 1]=1-e^{-c} \geq c/2$ (since $0 \leq c< 1$) 
and that $\pr[B_1=B_2] = \Omega (\min\{1, 1/\sqrt{b}\})$. To prove the latter lower bound, we will 
use the fact that $B_1, B_2$ are i.i.d. and that their 
common distribution $B$ is supported on integers and has standard deviation $\sigma = \sqrt{b}$. 
By Chebyshev's inequality, we have that $\pr\left[|B-b| = O(\sigma)\right] \geq 1/2$.
Since $B$ is has integer support, there exists a set of integers $S$ 
with cardinality $|S| \leq 1+ O(\sigma)$ such that $\pr[B \in S] \geq 1/2$. 
Now note that $\pr[B_1 = B_2] = \sum_{i \geq 0} \pr[B = i]^2 \geq \sum_{i \in S} \pr[B = i]^2 
\geq (1/|S|) \pr[B \in S]^2 \geq 1/(4|S|)$, where the second inequality follows by the convexity
of the quadratic function. Therefore, 
$\pr[B_1 = B_2] = \Omega(1/(1+O(\sigma))) = \Omega\left(\min\{1, 1/\sigma\} \right)$, as desired.

\medskip

\noindent {\bf Case II: $c \geq 1$.}

In this case, there exists a universal constant $\delta_0>0$ such that $\delta_0 = \pr[\min\{C_1,C_2\} \geq c/2]$. 
We will show that $\pr[|B_1-B_2| \leq c/4] = \Omega(\min\{1,c/\sqrt{b}\})$. Using 
\eqref{eq:Zi-lb}, the latter inequality implies that 
\begin{eqnarray*}
\E[Z_i] &\geq&  (c/4) \, \pr\left[\min\{C_1,C_2\} \geq c/2\right] \, \pr[|B_1-B_2| \leq c/4] \\ 
            &=& (c/4) \, \delta_0  \, \Omega(\min\{1,c/\sqrt{b}\}) \\
            &=& \Omega\left(\min\{c, c^2/\sqrt{b}\}\right) \;.
\end{eqnarray*}
To establish the desired upper bound on $\pr[|B_1-B_2| \leq c/4]$, 
we apply the argument from Case I for the random variables $B'_i = \lfloor B_i/(c/4) \rfloor$, $i=1,2$.
Note that the $B'_i$ is an integer-valued random variable with standard deviation $\sigma' = 1+ O(\sqrt{b}/c)$,
and therefore $\pr[B'_1=B'_2]  = \Omega(\min\{1, 1/\sigma'\}) = \Omega(\min\{1, c/\sqrt{b}\})$.
Finally, we note that $\pr[|B_1-B_2| \leq c/4] \geq \pr[B'_1=B'_2]$.
This completes Case II.

% old argument had two sub-cases.
%To bound the RHS of \eqref{eq:Zi-lb2}, we consider two sub-cases.

%First, if $b=O(c^2)$, the standard deviation of $\Poi(b)$ is $\sqrt{b} = O(c)$, and therefore 
%there is a universal constant $\delta_1>0$ such that 
%$\pr [|B_1-B_2|<c/4]  = \delta_1$.
%This implies that $\E[Z_i]=\Omega(c)$. 

%On the other hand, if $b=\Omega(c^2)$, \eqref{eq:Zi-lb2} gives that
%$\E[Z_i] \geq (c/4) \pr[|B_1-B_2| \leq c/4]$. 
%We will show that $\pr[|B_1-B_2| \leq c/4] = \Omega(c/(1+\sqrt{b}))$.
%and therefore that $\E[Z_i]  =  \Omega(c^2/(1+\sqrt{b}))$. 

Recall that $c=|k p_i- k q_i|$ by definition. 
The proof of Claim~\ref{cl:pointwise} is now complete.
\end{proof}

\paragraph{Proof of Lemma~\ref{lem:close-exp-gap} (ii).}
 
Suppose that $\dtv(p, q) \geq \eps$. For each bin $i \in [n]$, we assign $i$ to set $S_1, S_2, S_3$
if the $\min \left\{ |k p_i - k q_i|, |k p_i - k q_i|^2, \frac{|k p_i - k q_i|^2}{\sqrt{k p_i+k q_i}}\right\}$
is equal to $|k p_i - k q_i|$, $|k p_i - k q_i|^2$, or $\frac{|k p_i - k q_i|^2}{\sqrt{k p_i+k q_i}}$
respectively (breaking ties arbitrarily). This defines a partition of $[n]$ into three sets, $S_1, S_2, S_3$.
Since $\sum_{i=1}^n |p_i - q_i| \geq \eps/2$, for at least one $j \in \{1, 2, 3\}$ we have 
that $\sum_{i \in S_j} |p_i - q_i| \geq \eps/6$. In each of these three cases, we will use Claim~\ref{cl:pointwise}
to prove the desired expectation lower bound. 

\medskip

\noindent {\bf Case 1: $\sum_{i \in S_1} |p_i - q_i| \geq \eps/6$.}

In this case, we have that $\E[Z] = \sum_{i=1}^n \E[Z_i] \geq  \sum_{i \in S_1} \E[Z_i]  = 
\Omega (k) \sum_{i \in S_1} |p_i - q_i| = \Omega(\eps k)$.
Since $k$ is assumed to satisfy~\eqref{eqn:k-close-bound} and in particular we 
have that $k \geq C \log(1/\delta)/\eps^2$, it follows that $\E[Z]  = \Omega(\sqrt{k\log(1/\delta)})$, as desired.

\medskip

\noindent {\bf Case 2: $\sum_{i \in S_2} |p_i - q_i| \geq \eps/6$.}

In this case, we have that $\E[Z] = \sum_{i=1}^n \E[Z_i] \geq  \sum_{i \in S_2} \E[Z_i] 
= \Omega (k^2) \sum_{i \in S_2} |p_i - q_i|^2 = \Omega(k^2 \eps^2/n)$,
where the last inequality follows from Cauchy-Schwarz and the fact that $|S_2| \leq n$.
Since $k$ is assumed to satisfy  \eqref{eqn:k-close-bound} and in particular 
$k \geq C n^{2/3} \log^{1/3}(\delta)/\eps^{4/3}$, 
it follows that $\E[Z]  = \Omega(\sqrt{k\log(1/\delta)})$, as desired.

\medskip

\noindent {\bf Case 3: $\sum_{i \in S_3} |p_i - q_i| \geq \eps/6$.}
In this case, we can similarly write that 
$$\E[Z] = \sum_{i=1}^n \E[Z_i] \geq  \sum_{i \in S_3} \E[Z_i] = 
\Omega(k^{3/2}) \sum_{i \in S_3}  \frac{(p_i-q_i)^2}{(p_i+q_i)^{1/2}} = \Omega\left(k^{3/2} \eps^2/n^{1/2} \right) \;,$$
where the last bound follows from our assumption that $\sum_{i \in S_3} |p_i - q_i| \geq \eps/6$ 
and a careful application of the generalized Holder's inequality. Recall that for any triple of vectors 
$x, y, z \in \R^m$, we have that $\sum_i |x_i y_i z_i| \leq \|x\|_r \|y\|_s \|z\|_t$, where $1/r+1/s+1/t=1$. 
Using this fact, we can write 
$$\sum_{i \in S_3} |p_i - q_i| = \sum_{i \in S_3} \frac{|p_i-q_i|}{(p_i+q_i)^{1/4}} \; (p_i+q_i)^{1/4} \; 1
\leq \left( \sum_{i \in S_3} \frac{(p_i-q_i)^2}{(p_i+q_i)^{1/2}} \right)^{1/2} \left(\sum_{i \in S_3} (p_i+q_i) \right)^{1/4} 
\left(\sum_{i \in S_3} 1^4 \right)^{1/4} \;,$$
where we used $x = \left(\frac{|p_i - q_i|}{(p_i+q_i)^{1/4}}\right)_{i \in S_3}$, $y = ((p_i+q_i)^{1/4})_{i \in S_3}$, 
$z = (1)_{i \in S_3}$, and $r=2, s=t=4$. Since $\sum_{i \in S_3} (p_i+q_i) \leq 2$ and $|S_3| \leq n$, 
we get that
$\sum_{i \in S_3} \frac{(p_i-q_i)^2}{(p_i+q_i)^{1/2}} = \Omega(\eps^2/n^{1/2})$,
as desired.

We have thus shown that $\E[Z] = \Omega(k^{3/2} \eps^2/n^{1/2})$.
Since $k$ is assumed to satisfy  \eqref{eqn:k-close-bound} and in particular 
$k \geq C n^{1/2} \log^{1/2}(\delta)/\eps^{2}$, 
it follows that $\E[Z]  = \Omega(\sqrt{k\log(1/\delta)})$, as desired.

This completes the proof of Lemma~\ref{lem:close-exp-gap} (ii). \qed 

\subsection{Concentration of Test Statistic: Proof of Theorem~\ref{thm:close-alg}} \label{ssec:close-conc}

By Lemma~\ref{lem:close-exp-gap}, we have that $\E[Z]=0$ in the completeness case and 
$\E[Z]=\Omega(\sqrt{k\log(1/\delta)})$ in the soundness case respectively. Combined with Claim~\ref{claim:exp-close-z-tilde},
we have that in the completeness case $\E[\widetilde{Z}] = O(\sqrt{k})$ and in the soundness case
$\E[\widetilde{Z}] = \Omega(\sqrt{k \log(1/\delta)})$.

The random variable $\widetilde{Z}$ depends on $4k$ inputs: the
choice, for each of the $4k$ samples, of which distribution to be
drawn from and which coordinate to land in.  $\widetilde{Z}$ is
$2$-Lipschitz in these $4k$ inputs.  An application of McDiarmid's
inequality to $\widetilde{Z}$ gives that
\[
\pr \left[ \left| \widetilde{Z} - \E[\widetilde{Z}] \right|  \geq C' \sqrt{k \log(1/\delta)} \right] < 2 e^{-2\frac{(C'\sqrt{(k \log(1/\delta)})^2}{4k \cdot 4}} = 2\delta^{(C')^2/8} = 2\delta^{C'''}
\]
for some constant $C'''$.  If we apply the variable substitution $\delta \gets (\delta/2)^{1/C'''}$, the RHS above becomes $\delta$ and the number of samples only changes by a constant factor.
Therefore, our tester is correct with probability at least $1-\delta$, as desired.

\section{Sample-Optimal Independence Tester} \label{sec:indep-alg}

\subsection{Intuition and Setup}\label{ssec:setup}

\dnew{The goal in independence testing is to distinguish between $p$ and $q=p_x \times p_y$, i.e., 
the product of the marginal distributions of $p$ on the two coordinates. 
Unfortunately, we cannot simply use our closeness tester to solve this problem, 
as the sample complexity would contain an $(nm)^{2/3}\log^{1/3}(1/\delta)/\eps^{4/3}$ term, 
which is sub-optimal even for constant $\delta$. 
Instead, we must take advantage of the fact that $q$ is a product distribution.

This issue is solved in the large $\delta$ case in~\cite{DK16} by flattening. 
The idea is that the error in their test statistic can be reduced if $q$ is guaranteed 
to have small $\ell_2$ norm. To achieve this, we use flattening to split up the heavy bins. 
This can be done especially effectively for product distributions, as we can use samples 
to identify the heavy bins in the marginals rather than having to individually identify 
all of the heavy bins in the product.

To make a technique like this work in our context, there are several obstacles that must be overcome. 
The first is that we need to know how flattening can be used to improve the concentration bounds 
on our test statistic $Z$. To see why this might be the case, we note that any bins 
with only a single sample do not contribute to $Z$, and thus do not contribute to its variance. 
In fact, with some extra work we can prove stronger concentration bounds on $Z$ 
that depend on the number $N$ of non-isolated samples. 
As distributions with small $\ell_2$ norm will likely produce fewer non-isolated samples, 
this will hopefully improve our concentration bounds.

Unfortunately, while the basic flattening technique~\cite{DK16} works in the large $\delta$ regime, 
it does not work with high probability. To overcome this issue, we note that the goal 
of our flattening is actually not to produce a distribution with small $\ell_2$ norm, 
but to ensure that the number of collisions among the samples used to compute $Z$ is relatively small. 
For this we note that if we are given a fixed pool $S$ of samples from which we draw samples 
both for the purposes of flattening and for computing $Z$, it can be shown that no matter what $S$ is, 
there is always a good probability that the samples to compute $Z$ have few collisions. 
The overall strategy for our tester will be to take this fixed set of samples and repeatedly 
try different subdivisions into flattening and testing samples until we find one that works.

The most basic unit of our tester will be an algorithm called \textsc{BasicTest}, 
which runs one iteration of this strategy and returns one of ``YES'', ``NO'', or ``ABORT'', 
with the last meaning that our attempt at flattening has failed and needs to be repeated.}

\begin{comment}
\dnew{[TODO: Add some explanation at least about what BasicTest is supposed to be before just diving into this.]}
\new{Our algorithm will use a subroutine that has three possible outcomes: ``YES", ``NO" and ABORT. The algorithm aborts when some parameters are not in our desired range and we repeat the subroutine on the same set of samples. However, we show that this will happen with probability bounded away from $1$, which implies that there is only a constant factor overhead in the running time of the full algorithm compared to the subroutine. The full algorithm follows:}
\end{comment}

\begin{algorithm}\label{alg:Full-indep}
    \SetKwInOut{Input}{Input}
    \SetKwInOut{Output}{Output}
    \SetKwRepeat{Do}{do}{while}
\new{
    \Input{Sample access to a $2$-dimensional distribution $p$ over $[n]\times [m]$%, and $k=C\left(\frac{n^{2/3}m^{1/3}\log^{1/3}(1/\delta)}{\eps^{4/3}}+\frac{\sqrt{nm\log(1/\delta)}}{\eps^2}+\frac{\log(1/\delta)}{\eps^2}\right)$ where $C$ is a sufficiently large universal constant.
    }
    \Output{``YES" if $p\in\mathcal{P}$,  ``NO" if $\inf_{q\in\mathcal{P}} \dtv(p,q)\geq \eps$, where $\mathcal{P}$ is the set of product distributions, both with probability at least $1-\delta$.}
    $k\gets C \left( n^{2/3}\log^{1/3}(1/\delta)/\eps^{4/3}+ \big(n^{1/2} \log^{1/2}(1/\delta) + \log(1/\delta)\big)/\eps^2 \right)$,
    where $C>0$ is a sufficiently large universal constant.\\
    $S\gets 100k $ samples from $p$.\\
    $result\gets \textrm{ABORT}$\\
    \While{$result=\textrm{ABORT}$}{$result\gets \textsc{BasicTest}(S)$}
    return $result$

\caption{ \textsc{FullTest}($\overline{S}$): Given a distribution $p$ over $[n]\times [m]$ (where $n\geq m$), % with marginals $p_x,p_y$, test if $p_x$ and $p_y$ are independent.
test if $p$ is a product distribution.
}
}
\end{algorithm}

\new{

The main result of this section is the following theorem:

\begin{theorem}\label{thm:ub-main}
There exists a universal constant $C>0$ such that the following holds:
When
\begin{equation}
k \geq C \left( n^{2/3} m^{1/3} \log^{1/3}(1/\delta) / \eps^{\new{4/3}} + ((nm)^{1/2}\log^{1/2}(1/\delta)+\log(1/\delta))/\eps^2 \right) \;,
\end{equation}
Algorithm \textsc{FullTest} is an $(\eps, \delta)$-independence tester in total variation distance.
\end{theorem}

%The goal of this section is to prove the following theorem:
%\begin{theorem}[Sample Complexity Upper Bound for Independence testing]\label{thm:ub-main}
%There exists a computationally efficient $(\eps, \delta)$-independence tester for discrete distributions on $[n] \times [m]$,
%where $n \geq m$, with sample complexity
%$$\Theta\left( n^{2/3} m^{1/3} \log^{1/3}(1/\delta) / \eps^{\new{4/3}} + ((nm)^{1/2}\log^{1/2}(1/\delta)+\log(1/\delta))/\eps^2 \right) \;.$$
%\end{theorem}

}

\paragraph{Setup.}  Our independence testing procedure
\textsc{BasicTest} has the following basic structure:

\begin{enumerate}
\item Choose a large multiset set of samples $\overline{S}$.
\item Choose from $\overline{S}$ a flattening $F = (F_x, F_y)$, and possibly \textrm{ABORT}.
\item Choose from $\overline{S}$ a set $S$ of ``flattened'' samples, and possibly \textrm{ABORT}.
\item Use $S$ to compute a test statistic $Z$.
\item Accept or reject based on the test statistic.
\end{enumerate}

At various points in the process, the algorithm may choose to
\textrm{ABORT} (for example, if the number of non-singletons in $S$ is
$100\times$ more than expected).  We will show that \dnew{if the algorithm is run on a random set $\overline{S}$ of samples}, 
the probability of outputting a wrong answer is $O(\delta)$, \dnew{but that for \emph{any} set $\overline{S}$ of samples} the chance of
aborting is at most $1/2$.  Therefore, when we abort, 
we can start over from Step~2, and repeat until we output ``YES" or ``NO", 
without increasing the sample complexity and with only $O(\delta)$ failure probability.

\paragraph{Flattening.}  Flattening involves choosing a set $F$ of
samples from the distribution $p$ with marginals $p_x$ and $p_y$.  We
then flatten the rows and columns of $p$ independently, giving us a
new distribution $p^f$ with marginals $p^f_1$ and $p^f_2$. % This flattened distribution has the property XXX.
%\subsection{Flattening}\tnote{To be merged into section \ref{sec:prelim}} \textcolor{red}{(Complete: please proofread this subsection.)}
The following definition appears as Definition~2.4 in \cite{DK16} and describes 
a subdivision of the domain of a distribution $p$ 
that aims at reducing its $\ell_2$ norm. %However, it does make sure that the domain size increases by at most a constant factor, in order to minimize the effect on the sample complexity of the closeness testing part.
For this transformation to be useful to us, we need to always make sure 
that the domain size does not increase by more than a constant factor as a result.
\begin{definition}[\cite{DK16}]
Given a distribution $p$ on $[n]$ and a multiset $S$ of elements from  $[n]$, we define the \emph{split} distribution $p_S$ over $[n+|S|]$ as follows:  For $1\leq i\leq n$, let $f_i$ be the number of times element $i$ appears in $S$, and %$b_i$ be the number of heads of a biased coin, which is flipped every time element $i$ appears and  comes up heads with probability $d/k$.  Also, let
$a_i=1+f_i$.  Our new distribution $p_S$ is supported on the set $B=\{(i,j): i\in[n], 1\leq j \leq a_i\}$. In order to get a sample $(i,j)$ from $p_S$, we first draw $i$ according to $p$ and then $j$ uniformly at random from $[a_i]$.
\end{definition}

%\textbf{Remark:}
Note the following fact about split distributions:
\begin{fact}
Let $p$ and $q$ be probability distributions on $[n]$, and $S$ a given multiset of $[n]$. 
Then: (i) We can simulate a sample from $p_S$ or $q_S$ by taking a single sample from $p$ or $q$, respectively. 
(ii) It holds that $\|p_S-q_S\|_1= \|p-q\|_1$.
\end{fact}
When we are dealing with multidimensional distributions, it will be useful to have 
a definition of flattening only on a specific marginal.
The definition below is given for $2$-dimensional distributions, but it can be easily generalized.

\begin{definition}
Given a distribution $p$ on $[n]\times [m]$ with marginals $p_x$ and $p_y$. Also let $S$ be a multiset of elements from  $[n]$ (respectively $[m]$), we define the \emph{row-split} (respectively \emph{column-split}) distribution $p_S$ over $[n+|S|]\times [m]$ (respectively $[n]\times [m+|S|]$) as follows: in order to get a sample $((i,k),j)$ (respectively $(i,(j,k))$) from the \emph{row-split} (respectively \emph{column-split}) distribution $p_S$, 
we first draw $(i,j)$ according to $p$ and then independently draw $k$ uniformly at random from $[a_i]$ (respectively $[a_j]$).
% For $1\leq i\leq n$, let $f_i$ be the number of times element $i$ appears in $S$, and $b_i$ be the number of heads of a biased coin, which is flipped every time element $i$ appears and  comes up heads with probability $d/k$.  Also, let $a_i=1+b_i$.  Our new distribution $p_S$ is supported on the set $B=\{(i,j): i\in[n], 1\leq j \leq a_i\}$. In order to get a sample $(i,j)$ from $p_S$, we first draw $i$ according to $p$ and then $j$ uniformly at random from $[a_i]$.
\end{definition}

\paragraph{Test Statistic.}  Define the product distribution
$q^f := p^f_1 \times p^f_2$.  \dnew{Note that $\dtv(p^f,q^f)=\dtv(p,q)$ where $q=p_x\times p_y$. Therefore, the goal of determining whether $p$ is a product distribution or far from it is equivalent to distinguishing between $p^f=q^f$ and $p^f$ far from $q^f$.}
In addition to sampling from $p^f$, we can sample $q^f$ by taking two
samples from $p^f$: we combine the first coordinate of the first
sample with the second coordinate of the second sample.

The sample set $S$ consists of four pieces:
\begin{itemize}
\item $S_{p0}, S_{p1}$: two sets of $\Poi(k)$ samples from $p^f$.
\item $S_{q0}, S_{q1}$: two sets of $\Poi(k)$ samples from $q^f$.
\end{itemize}

We let $X^{(p0)}_{u}$ denote the number of times element $u$ appears in
$S_{p0}$, and similarly for the other three sets.  For each $u$ in the
range of $q^f$ we get the test statistic:

\[
  Z_u:=|X_u^{(p0)}-X_u^{(q0)}|+|X_u^{(p1)}-X_u^{(q1)}|-|X_u^{(p0)}-X_u^{(p1)}|-|X_u^{(q0)}-X_u^{(q1)}|
\]
Our final test statistic is the sum of this:
\[
  Z := \sum_u Z_u.
\]

Note that, if a given item $u$ appears exactly once in the entire set
$S$ of samples, then $Z_u = 0$.  We say that such a sample is a
\emph{singleton}, and define $N \leq \abs{S}$ to be the number of
non-singleton samples.

\subsection{Concentration of $Z$}

The goal of this section is to prove that the test statistic $Z$
concentrates.  We will show this happens for any setting of the
flattening $F$, and ignoring the possibility of \textrm{ABORT} (that
is, if we ran even aborted procedures to completion).  In particular, our goal is the following
lemma:
\begin{lemma}\label{lem:Zconc}
  For a fixed flattening $F$ and any $\delta > 0$, there exists a
  constant $C>0$ such that
  \[
    \pr[|Z-\E[Z]|>C\cdot\sqrt{(N+\log(1/\delta))\log(1/\delta)}]\leq \delta.
  \]
\end{lemma}
Intuitively, the idea is that since singletons do not change the
statistic, the variance---and concentration---of $Z$ should depend on
the number of non-singletons $N$ rather than the total number of
samples $k$.  Note that the concentration is relative to $N$, which is
also a random variable.

We show this using \emph{symmetrization}.  For the sake of analysis we
introduce an independent copy of the statistic $Z'$, generated from
another set $S'$ of samples.  Let $T = S \cup S'$ be the set of all
samples used by $Z$ and $Z'$, and let $M$ be the number of
non-singletons in $T$.

Note that we could generate these same variables in a different way:
rather than first generating $S_{p0}$ and $S'_{p0}$ with $\Poi(k)$
samples each and setting $T_{p0} = S_{p0} \cup S'_{p0}$, we can
instead first sample $T_{p0}$ with $\Poi(2k)$ samples, then randomly
assign each sample in $T_{p0}$ to one of $S_{p0}$ and $S'_{p0}$ (and
similarly for $p1, q0, q1$).  These are equivalent generative
processes.  This second process leads to the following lemma:

\begin{lemma}\label{lem:ZZ'}
  For every possible $T$, and any $\delta > 0$,
  \[
    \pr[|Z-Z'|> \sqrt{8M\log(2/\delta)} \mid T]\leq \delta.
  \]
\end{lemma}
\begin{proof}
  We apply McDiarmid's inequality, and use the alternative generative
  process.  Conditioned on $T = (T_{p0}, T_{p1}, T_{q0}, T_{q1})$, the
  only randomness lies in whether each sample $v$ is placed in $S$ or
  $S'$.  Let $c_v$ be the maximum amount that $|Z - Z'|$ can change by
  when $v \in T$ is switched between $S$ and $S'$.  Switching $v$ can
  only change $Z$ by at most $2$, and similarly for $Z'$, so
  $c_v \leq 4$.  Moreover, if $v$ is a singleton in $T$, then
  switching $v$ has zero effect on $Z$ or $Z'$, so $c_v = 0$.  Hence
  \[
    \sum_{v \in T} c_v^2 \leq 16 M.
  \]
  Since $Z$ and $Z'$ are identically distributed,
  $\E[Z - Z' \mid T] = 0$.  Therefore McDiarmid's inequality states
  that, for any $t$,
  \[
    \pr[| Z - Z'| \geq t \mid T] \leq 2 e^{-\frac{2t^2}{16M}}.
  \]
  Setting $t$ appropriately gives the result.
\end{proof}

Since our desired lemma is in terms of $N$, not $M$, we relate the
two:
\begin{lemma}\label{lem:NM}
  There exists a constant $C$ such that, for every possible $T$, and
  any $\delta > 0$,
  \[
    \pr[M > C(N + \log(1/\delta)) \mid T] \leq \delta.
  \]
\end{lemma}
\begin{proof}
  We again use the alternative generative process.  There are $M$
  non-singletons in $T$, which means we can pair them up into $M/2$
  disjoint pairs of colliding elements.  Each such pair has a $1/4$
  chance of having both elements land in $S$, independent of every
  other pair.  Let $n$ be the number of such pairs that land entirely
  in $S$.  By a Chernoff bound:
  \[
    \pr[n \leq M/16 \mid T] \leq e^{-M/C}
  \]
  for some constant $C \geq 8$.  Now, if $T$ is such that
  $M \leq C \log(1/\delta)$, the lemma statement is trivially true.
  Otherwise, since $N \geq 2n$,
  \[
    \pr[M \geq 8N \mid T] \leq \delta
  \]
  as desired.
\end{proof}

We also need to prove a constant-probability version of the result:
\begin{lemma}\label{lem:EZ}
It holds that \[
    \pr[ \abs{Z - \E[Z]} \geq C\sqrt{N+1}] \leq 1/2.
  \]
\end{lemma}
\begin{proof}
  We will show this with Markov's inequality, by showing
  \begin{align}\label{eq:EZ}
    \E[ (Z - \E[Z])^2 / (N+1)] = O(1)
  \end{align}
  using symmetrization.  Since $Z'$ is independent of $Z$, and by
  convexity,
  \begin{align}
    \E[(Z - \E[Z])^2 / (N+1)] &\leq \E[(Z - Z')^2 / (N+1)]\notag\\
    &= \E_T[ \E[(Z - Z')^2 / (N+1) | T] ].\label{eq:Zsym}
  \end{align}
  For any fixed $T$, by Lemma~\ref{lem:ZZ'} and Lemma~\ref{lem:NM} \new{applied with $\delta/2$} and
  a union bound we have with probability $1-\delta$ that both:
  \begin{align*}
    (Z - Z')^2 &\leq 8M \log(4/\delta)\\
    N &\geq M/C - \log(2/\delta)
  \end{align*}
  The latter equation implies $N+1 \geq M/(C \log(2/\delta))$%\tnote{Is this true for any $C,\delta$?}
  , and hence
  \[
    (Z - Z')^2 / (N+1) \leq 8C \log(4/\delta) \log(2/\delta)
  \]
  with probability $1-\delta$.  This strong concentration implies a
  bound in expectation: %\tnote{How does that follow?}
  \[
    \E[(Z - Z')^2 / (N+1) | T] \leq O(8C) = O(1).
  \]
  Plugging back into \eqref{eq:Zsym} gives~\eqref{eq:EZ}, which implies the result.
\end{proof}

We now have the tools for the main result of the section.
\begin{proof}[Proof of Lemma~\ref{lem:Zconc}]
  Consider any two thresholds $\tau$ and $\tau'$, where $\tau$ is a
  random variable depending on the sampling used for $Z$ and $\tau'$
  depends on that for $Z'$.  Because $Z'$ is independent of $Z$, we
  have:
  \begin{align*}
    \pr[|Z - \E[Z]| > \tau \cap |Z' - \E[Z]| < \tau'] = \pr[|Z - \E[Z]| > \tau] \pr[|Z' - \E[Z]| < \tau'].
  \end{align*}
  On the other hand,
  \begin{align*}
    \pr[|Z - \E[Z]| > \tau \cap |Z' - \E[Z]| < \tau'] %&\leq \pr[|Z - Z'| > \tau-\tau' \cap |Z' - \E[Z]| < \tau']\\
                                               &\leq \pr[|Z - Z'| > \tau-\tau'].
  \end{align*}
  Hence
  \begin{align}\label{eq:tau}
    \pr[|Z - \E[Z]| > \tau] \leq \pr[|Z - Z'| > \tau-\tau'] / \pr[|Z' - \E[Z]| < \tau'].
  \end{align}
  We now define these two thresholds $\tau$ and $\tau'$.

  \paragraph{Defining $\tau'$.} By Lemma~\ref{lem:EZ} applied to $Z'$,
  with $50\%$ probability we have
  \begin{align}\label{eq:tau2'}
    \abs{Z' - \E[Z]} \leq O(\sqrt{N'+1}).
  \end{align}
  Define $\tau'$ to be this RHS.

  By Lemma~\ref{lem:NM}, with $1-\delta$ probability we have
  \begin{align}
    M =O( N + \log(1/\delta)).\label{eq:NM}
  \end{align}
  (Note that we are no longer conditioning on $T$.)  Since
  $N' \leq M$, this implies that there exists a constant $C>0$%\tnote{We should say this is a bit different than the one in Lemma \ref{lem:NM} since we are loosing the +1}
  such that
  \begin{align}\label{eq:tau'}
    \tau' \leq C \sqrt{N + \log(1/\delta)}
  \end{align}
  with probability $1-\delta$.

  \paragraph{Defining $\tau$.}
  On the other hand, combining~\eqref{eq:NM} with Lemma~\ref{lem:ZZ'},
  with $1-2\delta$ probability we have
  \[
    |Z-Z'| \leq O(\sqrt{(N + \log(1/\delta))\log(2/\delta)}).
  \]
  We would like to define $\tau$ to be this RHS plus $\tau'$, but this
  would be invalid: $\tau$ must be independent of $Z'$.  Hence we
  instead define $\tau$ to be this RHS plus
  $C\sqrt{N + \log(1/\delta)}$; by~\eqref{eq:tau'}, this is larger
  than the RHS plus $\tau'$ with $1-\delta$ probability.  Hence:
  \begin{align}\label{eq:tautau'}
    \pr[|Z-Z'| > \tau - \tau'] \leq 3\delta
  \end{align}
  for this $\tau$, which is $O(\sqrt{(N + \log(1/\delta))\log(2/\delta)})$.

  \paragraph{Combining the results.} Plugging~\eqref{eq:tautau'}
  and~\eqref{eq:tau2'} into~\eqref{eq:tau}, we have for this $\tau$ that
  \[
    \pr[|Z - \E[Z]| > \tau] \leq 3\delta / (1/2) = 6 \delta.
  \]
\new{ Using $\delta^\prime=\delta/6$ gives the desired result.}
  %Rescaling $\delta$ gives the result.
\end{proof}

\subsection{Algorithm}

We begin with a helper algorithm \textsc{BasicTest} (i.e., Algorithm \ref{alg:Basic-indep}): %(called \textsc{BasicTest}(S) or something, don't know how to get the name to show up).

\begin{algorithm}\label{alg:Basic-indep}
    \SetKwInOut{Input}{Input}
    \SetKwInOut{Output}{Output}

    \Input{A Multiset $\overline{S}$ of $100k$ samples from $[n]\times [m]$ with $k=C\left(\frac{n^{2/3}m^{1/3}\log^{1/3}(1/\delta)}{\eps^{4/3}}+\frac{\sqrt{nm\log(1/\delta)}}{\eps^2}+\frac{\log(1/\delta)}{\eps^2}\right)$, where $C$ is a sufficiently large universal constant.}
    \Output{Information relating to whether these samples came from an independent distribution.}
    \tcc{Choose flattening $F$}
    $F_x\gets \emptyset$\\
  $F_y\gets \emptyset$\\
  \For{$s \in \overline{S}$ }{
    $F_x=F_x\cup \{s\}$ with prob $\min\{n/100k,1/100\}$\\
    $F_y=F_y\cup \{s\}$ with prob $m/100k$ \tcp*{note that $k>m$ always.}
    }

    \If{\label{flattening}$|F_x|>10n$ or $|F_y|>10m$}{\label{large F} return \textrm{ABORT}}

    \tcc{Draw samples $S_p^f, S_q^f$}
Let $\overline{S}^\prime=\{(x_i,y_i)\}$ be a uniformly random permutation of $\overline{S}\setminus (F_x\cup F_y)$\\
Draw $\ell,\ell^\prime\sim \Poi(2k)$.\label{perm}\\
\If{\label{sampling}$2\ell+\ell^\prime>|\overline{S}^\prime|$ }{\label{large l} return \textrm{ABORT}}
Let $S_q=\{(x_{2j-1},y_{2j})\}_{j=1}^\ell , S_p=\{(x_j,y_j)\}_{j=2\ell+1}^{2\ell+\ell^\prime}$\\
%Compute the number of sub-bins for row-flattening and column-flattening using $F_x,F_y$\\
Create $S_p^f,S_q^f$ by assigning to corresponding sub-bins uniformly at random\\
   Let $N_q:$ $\sharp$samples in $S_q^f$ that collide with another sample in $S_p^f\cup S_q^f$.\\
   Let $N_p:$ $\sharp$samples in $S_p^f$ that collide with another sample in $S_p^f$\\

\If{$N_q > c \max(k/m,k^2/mn)$}{\label{large Nq} return \textrm{ABORT}}
  \If(\tcp*[f]{$C'$ a sufficiently large constant}){$N_p > 20N_q+ C'\log(1/\delta)$ \label{cond:large N}}
  {return ``NO"\label{line:reject0}}
  \tcc{Compute test statistic $Z$}

  % Draw a multiset $S_q\subseteq [n], |S|\sim \Poi(k)$ samples from distribution $q$\\
   Flag each sample of $S_p^f,S_q^f$ independently with probability $1/2$.\\
  % \textcolor{red}{Themis: Unify notation with closeness tester.}\\
   Let  $X^{(p0)}_i,X^{(q0)}_i$ be the counts for the number of times element $i$ appears \emph{flagged} in each set $S_p^f,S_q^f$ respectively and $X^{(p1)}_i,X^{(q1)}_i$ be the corresponding counts on \emph{unflagged} samples.\\
   \label{statistic}Compute the statistic $Z=\sum_{i} Z_i$, where $Z_i=|X_i^{(p0)}-X_i^{(q0)}|+|X_i^{(p1)}-X_i^{(q1)}|-|X_i^{(p0)}-X_i^{(p1)}|-|X_i^{(q0)}-X_i^{(q1)}|$.  \\
   \eIf{\new{$Z< C^\prime\cdot \sqrt{\dnew{\min(k,(k^2/(mn)+k/m))}\log(1/\delta)}$}}{return ``YES"\label{line:accept}}{return ``NO"\label{line:reject1}}
\caption{ \textsc{BasicTest}($\overline{S}$): Given a distribution $p$ over $[n]\times [m]$ (where $n\geq m$) % with marginals $p_x,p_y$, test if $p_x$ and $p_y$ are independent.
test if $p$ is a product distribution.}
\end{algorithm}

\dnew{Our analysis will depend on two key facts:
\begin{enumerate}
\item For \emph{any} set of samples $\overline{S}$, the probability that \textsc{BasicTest} returns ABORT is at most $1/2$.
\item If \textsc{BasicTest} is run on a set of i.i.d. samples from $p$, the probability that it returns an incorrect answer (``NO'' if $p$ is actually independent, or ``YES'' if $p$ is $\eps$-far from independent) is at most $\delta$.
\end{enumerate}
The latter of these points will hold because our algorithm will ABORT unless $N_q$ is small. This, 
combined with \new{Lemma \ref{lem:Zconc} and Claim \ref{clm: exp gap}}, 
will imply that the output is correct (along with a separate argument (see Lemma \new{\ref{lem: large Np}}) 
for when $N\gg N_q$).

To show the first of these points, one can first use Markov to bound the probability of aborting due to $F_x$ or $F_y$ or $\ell$ or $\ell'$ being too large. The more interesting case is to show that $N_q$ is bounded with appropriate probability. This will follow from the following lemma:
}

\begin{lemma}\label{lem:ENq}
  For any set of samples $\overline{S}$,
  \[
    \E[N_q \mid \overline{S}] =O\left( \max\left(\frac{k^2}{nm}, k/m\right)\right),
  \]
  \dnew{where $N_q$ is considered to be $0$ in the case that the algorithm aborts before computing it.}
\end{lemma}
\begin{proof}%\textcolor{red}{[XXX: Check this.]}

\dnew{Throughout this proof we will condition on $\overline{S}$. We note that $\ell\geq k/2$ except with probability exponentially small in $k$, \new{in which }%and that in this
 case $N_q=O(k)$. Thus, the contribution from the case where $\ell < k/2$ is $O(1)$ and we can henceforth assume that $\ell \geq k/2$ (note that given the size of the parameters $k/m > 1$).

In order to bound $N_q$ we bound it as a sum of simpler random variables whose expectations we can bound individually. For $1\leq i\leq 100k$, we let $N_i$ be $0$ unless the $i^{th}$ element of $\overline{S}$ is in $S_p$, and in that case, it is the number of elements of $S_q^f$ that the corresponding element of $S_p^f$ collides with (with the exception that we define $N_i$ to be $0$ if $\ell < k/2$). For $1\leq i \neq j \leq 100k$ let $N_{i,j}$ be $0$ unless one of the elements of $S_q$ is obtained by taking the $x$-coordinate from the $i^{th}$ element of $\overline{S}$ and $y$-coordinate from the $j^{th}$ element of $\overline{S}$, and if so is equal to the number of other elements in $S_q^f$ that the corresponding element of $S_q^f$ collides with (with the exception that we define $N_{i,j}$ to be $0$ if $\ell < k/2$). It is easy to see that
\begin{equation}\label{NqBoundSumEqn}
N_q \leq \sum_i N_i + \sum_{i,j} N_{i,j}.
\end{equation}
Our final result will follow from two bounds:
Firstly, for all $i$, we claim that
\begin{equation}\label{NiBoundEqn}
\E[N_i] = O(\max\left(\frac{k^2}{nm}, k/m\right)/k).
\end{equation}
We also claim that for all $i,j$ that
\begin{equation}\label{NijBoundEqn}
\E[N_{i,j}] = O(\max\left(\frac{k^2}{nm}, k/m\right)/k^2).
\end{equation}

We begin with our proof of Equation \eqref{NiBoundEqn} as it is slightly easier. Assume that the $i^{th}$ element of $\overline{S}$ is $(X,Y)$. Let $C_X$ denote the number of other elements of $\overline{S}$ with the same $x$-coordinate and $C_Y$ the number with the same $y$-coordinate. Upon flattening, let $F_X$ and $F_Y$ denote the number elements of $F_x$ equal to $X$ and the number of elements of $F_y$ equal to $Y$, respectively. Note that $F_X$ is distributed as a binomial distribution $\Bin(C_X,\min(n/100k,1/100))$ and thus $\E[1/(F_X+1)] = O(1/(C_X \min(n/k,1)))$. Similarly, $\E[1/(F_Y+1)] = O(k/(C_Y m))$.

Once we have conditioned on the flattening sets $F_x$ and $F_y$, we consider $C_{X,Y}$, the number of elements of $S_q$ equal to $(X,Y)$, where $C_{XY}$ is set to $0$ if $\ell < k/2$ (recall that this case can safely be ignored in our final analysis). We claim that $\E[C_{X,Y}|F_x,F_y] \ll C_XC_Y/k$. This is because the expectation of $C_{X,Y}$ is a sum of all pairs of one of the $C_X$ elements of $S$ with the correct $x$-coordinate and one of the $C_Y$ elements of $S$ with the correct $y$-coordinate of the probability that this pair of elements is used to create an element of $S_q$. We claim that this probability is $O(1/k)$. In fact, this probability is at most $1/\ell$\new{, where $\ell\geq k/2$ due to our conditioning}. That is because even conditioning on $\ell$ and which $2\ell$ elements of $S$ are used to construct the elements of $S_q$, there is only an $O(1/\ell)\new{=O(1/k)}$ probability that the two designated elements of $S$ are adjacent to each other after the random permutation is applied.

However, once $S_q$ is fixed, each of these $C_{X,Y}$ elements that might collide with our $i^{th}$ element of $S$ only do if they are mapped to the same sub-bin. This happens only with probability $1/((1+F_X)(1+F_Y))$. Therefore, we have that:
$$
\E[N_i|C_{X,Y},F_X,F_Y]  = \frac{C_{X,Y}}{(1+F_X)(1+F_Y)}.
$$
Therefore, \new{using the fact that $F_X,F_Y$ are independent random variables,} we have that
\begin{align*}
\E[N_i] & \leq \sup_{F_X,F_Y}(\E[C_{X,Y}|F_X,F_Y])\E[1/(1+F_X)]\E[1/(1+F_Y)]\\ & = O(C_XC_Y/k)O(\max(k/n,1)/C_X)O(k/(C_Y m)) = O(\max(k/(mn),1/m)),
\end{align*}
as desired.

The proof of Equation \eqref{NijBoundEqn} is similar. Assume that the $i^{th}$ element of $\overline{S}$ has $x$-coordinate $X$ and that the $j^{th}$ element has $y$-coordinate $Y$. Let $C_X$ and $C_Y$ be the number of other elements of $\overline{S}$ with $x$-coordinate equal to $X$ and $y$-coordinate equal to $Y$, respectively. Again let $F_X$ and $F_Y$ denote the number elements of $F_x$ equal to $X$ and the number of elements of $F_y$ equal to $Y$, respectively. Once again $\E[1/(F_X+1)] = O(1/(C_X \min(n/k,1)))$ and $\E[1/(F_Y+1)] = O(k/(C_Y m))$.

We now let $C_{X,Y}$ be $0$ unless $\ell \geq k/2$ \emph{and} the $i^{th}$ and $j^{th}$ elements pair to make an element of $S_q$, and in this case define it to be the number of other elements of $S_q$ equal to $(X,Y)$. We claim now that $\E[C_{X,Y}|F_x,F_y] \ll C_XC_Y/k^2$ (note that this differs from the above because of the $k^2$ in the denominator rather than $k$). This is because $C_{X,Y}$ is the sum over the $C_XC_Y$ pairs of other elements with the correct $x$ and $y$ values of the probability that this pair of elements of $S$ \emph{and} the pair of the $i^{th}$ and $j^{th}$ elements both end up in $S_q$. Even conditioning on $F_x,F_y$ and $\ell$, the probability that the random permutation of elements put \new{the two elements of} both of these pairs next to each other is $O(1/\ell^2) = O(1/k^2)$. Thus, $\E[C_{X,Y}|F_x,F_y] \ll C_XC_Y/k^2$.

From here the argument is the same as above. Each of these $C_{X,Y}$ elements of $S_q$ has only a $1/((F_X+1)(F_Y+1))$ of colliding with our designated one after assigning them to random sub-bins. Thus, we have that
Therefore, we have that:
$$
\E[N_{i,j}|C_{X,Y},F_X,F_Y]  = \frac{C_{X,Y}}{(1+F_X)(1+F_Y)}.
$$
And thus,
\begin{align*}
\E[N_{i,j}] & \leq \sup_{F_X,F_Y}(\E[C_{X,Y}|F_X,F_Y])\E[1/(1+F_X)]\E[1/(1+F_Y)]\\ & = O(C_XC_Y/k^2)O(\max(k/n,1)/C_X)O(k/(C_Y m)) = O(\max(1/(mn),1/(km))),
\end{align*}
as desired.

Our lemma now follows from combining Equations \eqref{NqBoundSumEqn}, \eqref{NiBoundEqn} and \eqref{NijBoundEqn}.
}
\end{proof}

\dnew{We are now prepared to prove the second of our main points about \textsc{BasicTest}.}

\begin{lemma}\label{lm:abort}
For any sample multiset $\overline{S}$, the probability that \textsc{BasicTest} returns \textrm{ABORT} is at most $1/2$.
\end{lemma}
\begin{proof}
\begin{comment}
\sketch{(sketch) The probability that $F_x$ or $F_y$ too big or $\ell+\ell'$ too big are easily seen to be negligible.}%\newline \textcolor{red}{TODO}\newline

\sketch{Otherwise, we just need to bound the expected size of $N_q$ and use Markov. For this, for each $(x,y)\in \overline{S}$, we bound the probability that $(x,y)\in S_p$, and the conditional expectation of the number of elements of $S_q$ it collides with. Similarly, for $(x,y),(x',y')\in \overline{S}$ we bound the probability that $(x,y')\in S_q$ and if so, bound the conditional probability that it collides with some other element of $S_q$. Both of these conditional probabilities should be $O(k/mn+1/m)$. Summing over the $k$ elements in $S_p\cup S_q$ gives the bound.}\newline
\end{comment}

First, consider the case that \textsc{BasicTest} returns \textrm{ABORT} in line \ref{large F}, because either $|F_x|>10n$ or $|F_y|>10n$.
Note that $F_x\sim \Bin(100k,\min\{n/100k,1/100\})$ and $F_y\sim \Bin(100k,m/100k)$. Therefore, we have that: $\E[|F_x|]\leq n$ and $\E[|F_y|]= m$.  By applying Markov's inequality for each random variable and a union bound, we get that
\[
\pr[(|F_x|>10n) \vee (|F_y|>10n)]\leq 1/5 \;.
\]
The second possibility to return \textrm{ABORT} is in line \ref{large l} when $2\ell+\ell^\prime>|\overline{S}^\prime|\geq 100k-|F_x|-|F_y|$.
Thus, we need to bound: $\pr[2\ell+\ell^\prime +|F_x|+|F_y|>100k]$. Note that by linearity of expectation:
\[\E[2\ell+\ell^\prime +|F_x|+|F_y|]=\E[2\ell]+\E[\ell^\prime] +\E[|F_x|]+\E[|F_y|]\leq 4k+2k+k+k=8k \;. \]
By applying Markov's inequality again, we get that:
\[
\pr[2\ell+\ell^\prime +|F_x|+|F_y|>100k]\leq 8/100 \;.
\]
It remains to bound the chance of \textrm{ABORT} on line~\ref{large Nq}.  By
Lemma~\ref{lem:ENq} and Markov's inequality,
\[
\pr[N_q>c\cdot \max\{k/m,k^2/nm\}|\overline{S}]<1/5 \;,
\]
for some constant $c$.

Using a union bound for all the above three cases, we get that the
probability that \textsc{BasicTest} returns \textrm{ABORT} is at most
$1/5+8/100+1/5<1/2$.
\end{proof}

For the rest of the analysis, we consider running \textsc{BasicTest} on a set $\overline{S}$ of random samples from some distribution $p$ on $[n]\times[m]$. We note that we can simulate the algorithm in the following way: First, for each $i$ from $1$ to $100k$, if our algorithm wants to add an element to $F_x$ or $F_y$, we generate a random element from $p$ and add it to the appropriate set(s). \dnew{If either $|F_x|>10n$ or $|F_y|>10m$ we abort, so we will condition on $F_x$ and $F_y$ for which this does not happen.} Next, we generate an infinite sequence of elements $(x_i,y_i)$ from $p$, and let $S_q$ be the set of $(x_{2j-1},y_{2j})$ for $j\in [1,\ell]$   and $S_p$ the set of $(x_j,y_j)$ for $j\in[2\ell+1,2\ell+\ell^\prime]$. Note that conditioned on not returning \textrm{ABORT}, this gives sets $F_x,F_y,S_p,S_q$ identically distributed as \textsc{BasicTest}. However, unconditionally, it gives an $S_q$ and $S_p$ sets of $\Poi(2k)$ samples from $q:=p_x\times p_y$ and $p$, respectively. \dnew{Furthermore, we can compute $Z,N_q$ and $N$ regardless. Note that this statistic $Z$ will be an instance of the statistic computed for our closeness tester applied to the distribution $p^f$ and $q^f$. In particular, Lemmas \new{\ref{lem:Zconc}, \ref{lem:ZZ'}, \ref{lem:EZ} and \ref{lem:ZNpNq}} will still apply to it.}

\dnew{For the next several lemmas, we consider $F_x$ and $F_y$ as being fixed and $Z,N_q$ and $N$ being computed in this way regardless of potential aborts. In the next few lemmas, we wish to show that with high probability $N$ will be $O(N_q)$ if $p$ is a product distribution. This will allow us to use our bounds on $N_q$ 
as bounds on $N$ (or more precisely, allow us to reject if $N$ is not bounded in terms of $N_q$).}

\begin{lemma}\label{lem:ZNpNq}
  For a fixed set of samples $S_p^f, S_q^f$, consider the distribution
  of $Z$ over the partition into $p0/p1$ and $q0/q1$.  We have:
  \[
    \pr[Z < N_p/6 - 2 N_q - 100] < 1/2.
  \]
\end{lemma}
\begin{proof}
  Let $X_i^{(p)}$ denote the number of times element $i$ appears in
  $S_p^f$, so that
  \[
    N_p = \sum_{i:X_i^{(p)} > 1}  X_i^{(p)}.
  \]
  Define the statistic $\widetilde{Z} = \sum_i \widetilde{Z}_i$, where
  \[
    \widetilde{Z_i} = X_i^{(p0)} + X_i^{(p1)} - \abs{X_i^{(p0)} - X_i^{(p1)}} = 2 \min(X_i^{(p0)},  X_i^{(p1)}),
  \]
  to be the value $Z$ would take if $S_q^f$ were empty.  Since $Z$ is
  $2$-Lipschitz and invariant to singletons, we have
  \begin{align}
    \abs{Z - \widetilde{Z}} \leq 2N_q.\label{eq:ZwtZ}
  \end{align}
  Hence, our goal is to show that $\widetilde{Z}$ is usually at least $N_p/4$.
  We have that 
  \[
    \E[\widetilde{Z_i}] \geq \floor{X_i^{(p)}/2} \;,
  \]
  because we can partition the elements $i$ into $\floor{X_i^{(p)}/2}$
  pairs, each of which has a $1/2$ chance of being divided between
  $p0$ and $p1$, and hence contributing $1$ to each of $X_i^{(p0)}$
  and $X_i^{(p1)}$, or $2$ to $\widetilde{Z_i}$.  We also have that
  \[
    \var(\widetilde{Z_i}) \leq 4X_i^{(p)} \;,
  \]
  because $\widetilde{Z_i}$ is a $2$-Lipschitz function of $X_i^{(p)}$
  independent random choices, and of course $\var(\widetilde{Z_i}) = 0$ if
  $X_i^{(p)} = 0$.  Therefore,
  \begin{align*}
    \E[\widetilde{Z}] \geq N_p / 3, \quad
    \var[\widetilde{Z}] \leq 4N_p \;.
  \end{align*}
  By Chebyshev's inequality, this means
  \[
    \pr[\widetilde{Z} < N_p/3 - 4 \sqrt{N_p}]  \leq 1/4,\quad \mbox{or }  \quad \pr[\widetilde{Z} < N_p/6 \text{ and } N_p > 600]  \leq 1/4.
  \]
  Combined with~\eqref{eq:ZwtZ}, we have
  \[
    \pr[Z < N_p/6 - 2N_q \text{ and } N_p > 600]  \leq 1/4 \;.
  \]
  But, of course, $\pr[Z < 0] = 0$, so for all $N_p$ we have that
  \[
    \pr[Z < N_p/6 - 2 N_q - 100] \leq 1/4 < 1/2 \;.
  \]
\end{proof}

\dnew{We can now bound the probability that we reject incorrectly on line \ref{line:reject0}.}

\begin{lemma}\label{lem: large Np}
If $p$ is a product distribution, then the probability that \textsc{BasicTest} returns ``NO" on line~\ref{line:reject0} is $O(\delta)$.
\end{lemma}
\dnew{This is essentially because if $N$ is a sufficient multiple of $N_q$ then by Lemma \ref{lem:ZNpNq} we have that $Z$ is likely to be at least a large multiple of $N$. However Lemma \ref{lem:close-exp-gap} says that $\E[Z]=0$ and Lemma \ref{lem:Zconc} says that $|Z-\E[Z]| \ll \sqrt{N\log(1/\delta)}$ with high probability.}
\begin{proof}
\begin{comment}
  \sketch{(Sketch) Generating unconditional samples as described above, we compute statistics $Z$ and $N$. We note that since $p=q$ $\E[Z]=0$. Therefore, we have that $\pr(Z>N/10 \wedge N>10\log(1/\delta)) =O(\delta)$ (or whatever) by a previous Lemma. However, if we reject $N_p > 10(N_q+\log(1/\delta))$. It is not hard to see that if this holds, that even conditioning on $S_p$ and $S_q$, $Z \gg N$ with high probability. Combining with the above gives our bound.}\newline
\end{comment}

  In the completeness case (i.e., $p$ is independent and $p=q$), we have
  that $\E[Z]=0$ due to symmetry (see Lemma~\ref{lem:close-exp-gap}). 
  Using this fact and Lemma \ref{lem:Zconc}, we have for some constant $C$ and
  $\tau_1(N) := C\sqrt{(N+\log(1/\delta))\log(2/\delta)}$, it holds
  \[
    \pr[Z>\tau_1(N)]\leq \delta/2 \;.
  \]
  On the other hand, Lemma~\ref{lem:ZNpNq} says that, for any $N_p$
  and $N_q$, and $\tau_2(N_p, N_q) := N_p/6 - 2 N_q - 100$, that
  \[
    \pr[Z \geq \tau_2(N_p, N_q) \mid N_p, N_q] \geq 1/2 \;,
  \]
  and so
  \begin{align*}
    \pr[Z>\tau_1(N)] &\geq \pr[Z > \tau_2(N_p, N_q) \cap \tau_2(N_p, N_q) \geq \tau_1(N)]\\
    &\geq \pr[\tau_2(N_p, N_q) \geq \tau_1(N)] \min_{N_p,N_q,N} \pr[Z > \tau_2(N_p, N_q) \mid N_p, N_q, N]\\
    &\geq \frac{1}{2}\pr[\tau_2(N_p, N_q) \geq \tau_1(N)] \;.
  \end{align*}
  Combining these observations gives that
  \begin{align}
    \pr[\tau_1(N) \leq \tau_2(N_p, N_q)] \leq \delta \;.
  \end{align}
  We claim that whenever we return ``NO" on line~\ref{line:reject0},
  condition $\tau_1(N) \leq \tau_2(N_p, N_q)$. The lemma statement
  follows.

  Note that, if $N_p \geq 20N_q + 1000$, we have
  \[
    \tau_2(N_p, N_q) = N_p/6 - 2 N_q - 100 \geq N_p / 20.
  \]
  Now $N$ is the number of non-singletons in $S_p^f \cup S_q^f$,
  which is $N_p + N_q$ plus the number of samples in $S_p^f$ that
  collide with a sample in $S_q^f$ but not with one in $S_p^f$. This latter is also at most $N_q$, so
  $N_p + N_q\leq N \leq N_p + 2N_q$.

  Therefore, if $N_p > 20 N_q + C'(\log(2/\delta))$, for any $C' > 10$
  we have that
  \[
    \frac{1}{C}\tau_1(N) = \sqrt{(N+\log(1/\delta))\log(2/\delta)} \leq \sqrt{1.1 \cdot N_p
      \cdot N_p/C'} \leq N_p \sqrt{1.1/C'} \;.
  \]
  Therefore, for $C' \geq 440 C^2$, we have
  \[
    \tau_1(N) \leq N_p / 20 \;.
  \]
  All combined, this means
  \[
    \pr[N_p > \max(20N_q + 1000, 20 N_q + C'(\log(2/\delta)))] \leq\pr[\tau_1(N) \leq \tau_2(N_p, N_q)] \leq \delta \;,
  \]
  for $C' \geq \max(440 C^2, 20)$.  Since we only return ``NO" on
  line~\ref{line:reject0} if this condition holds, the probability is
  at most $\delta$ as required.
\end{proof}

\dnew{Finally, we can bound the probability of \textsc{BasicTest} giving an incorrect output on lines \ref{line:reject1} or \ref{line:accept}.}

\begin{lemma}
If $p$ is a product distribution, then the probability that \textsc{BasicTest} returns ``NO" on line~\ref{line:reject1} is $O(\delta)$. Similarly, if $\dtv(p,q)>\eps$ then the probability that \textsc{BasicTest} returns ``YES" is $O(\delta)$.
\end{lemma}
\dnew{This holds because if we reach this stage of the algorithm $N=N_p+N_q$. We know by previous checks that $N_q$ is not too large and $N_p$ is not much bigger than $N_q$. This gives us strong concentration bounds on $Z$ and a careful analysis of the separation in expectations between the soundness and completeness cases will yield our result.}
\begin{proof}
\begin{comment}
\sketch{(Sketch) In this case we have that $N = O(k^2/nm+k/m)$. By a previous Lemma, this implies that except with probability $\delta$ that
$$\pr(|Z-\E[Z|F_x,F_y]|> C'\sqrt{(k^2/mn+k/m)\log(1/\delta)}\wedge N = O(k^2/nm+k/m)) < \delta.  $$
If $p$ is independent $\E[Z|F_x,F_y]=0$ and so we are done. If $\dtv(p,q)>\eps$, we need to show that $\E[Z|F_x,F_y]>2C'\sqrt{(k^2/mn+k/m)\log(1/\delta)}. $ However, we note that $p^f$ and $q^f$ are supported on sets of size $O(mn)$ and $\dtv(p^f,q^f)>\eps$, so using our lower bounds on $\E[Z]$ (and proving a new lemma about them), we are done.}\newline
\end{comment}

In order for the algorithm to return ``YES" or ``NO" on line~\ref{line:reject1}, 
it has to avoid ``aborting" or returning ``NO" on line~\ref{line:reject0}. 
Therefore, it must be the case that $N_p \leq 20N_q+ C'\log(1/\delta)$ and 
$N_q= O(\frac{k^2}{nm}+k/m)$, which implies $N= O(\frac{k^2}{nm}+k/m)$. 
Note also that $\frac{k^2}{nm} \geq \log(1/\delta)$ \new{,as well as $k\geq \log(1/\delta)$}
 by definition of $k$. \dnew{Note also the trivial bound that $N=O(k).$}

%\tnote{I guess this is due to our conditioning in a previous lemma that $\ell>k/2$ (and similarly $\ell^\prime>k/2$). Should we mention it again here? I don't think we have conditioned on $\ell^\prime>k/2$ previously. }

Therefore, by Lemma \ref{lem:Zconc}, we have that:
\[
\pr[|Z-\E[Z]|>C\sqrt{\dnew{\min(k,(k^2/mn+k/m))}\log(1/\delta})]\leq \delta/2 \;,
\]
for some constant $C>0$.

If $p$ is a product distribution (i.e., $p=q$), then by Lemma \ref{lem:close-exp-gap}, we have that $\E[Z]=0$. 
Thus, the algorithm will return ``NO" with probability at most $\delta/2$.

For the soundness case, where $\dtv(p,q)>\eps$,  it suffices to show the following lower bound on the expected value of $Z$:
\begin{claim}\label{clm: exp gap}
If $\dtv(p,q)>\eps$, then
\[
\E[Z]\geq 2C'\sqrt{\dnew{\min(k,(k^2/mn+k/m))}\log(1/\delta)} \;.
\]
\end{claim}
\begin{proof}
%\new{Recall that $N=N_p+N_q$ by definition.}
Suppose that we condition on the flattening samples. 
This will determine the flattened distributions $p^f,q^f$. From the proof of Lemma \ref{lem:close-exp-gap} it follows that:
\[
\E[Z|F_x,F_y]=\Omega\left(\min\left\{ \eps k, \frac{k^2\eps^2}{|D_{p^f}|}, \frac{k^{3/2}\eps^2}{\sqrt{|D_{p^f}|}} \right\}\right) \;,
\]
%For the regime of parameters we are interested in, the second term dominates and we get that: $\E[Z|F_x,F_y]=\Omega(\frac{k^2\eps^2}{|D_{p^f}|})$
where $|D_{p^f}|=\Theta(nm)$ is the domain size of the flattened distribution.
We now distinguish the following three cases:%\footnote{For cases 1 and 3, the $\Omega\left( \sqrt{\frac{k^2}{nm}\log(1/\delta)}\right)$ is missing. But, maybe we could argue that when they are relevant, one of the two sample complexity terms that are common in the closeness and independence testers, dominates.}
\begin{itemize}
\item \textbf{Case 1:} $\E[Z|F_x,F_y]=\Omega(\eps k)$. 

Using the fact that $k=\Omega(\log(1/\delta)/\eps^2)$, it follows that $\E[Z|F_x,F_y]=\Omega(\sqrt{k\log(1/\delta)})$. %$\Omega\left( \sqrt{(k/m) \log(1/\delta)}\right)$.
\item \textbf{Case 2:}  $\E[Z|F_x,F_y]=\Omega\left(\frac{k^2\eps^2}{nm}\right)$

\begin{itemize}
\item Using the fact that $k=\Omega\left(\frac{\sqrt{nm\log(1/\delta)}}{\eps^2}\right)$, we get that:
\[
\E[Z|F_x,F_y]=\Omega\left(\frac{k\eps^2}{nm}\cdot\frac{\sqrt{nm\log(1/\delta)}}{\eps^2}\right)=\Omega\left( \sqrt{\frac{k^2}{nm}\log(1/\delta)}\right).
\]
\item  Using the fact that $k=\Omega\left(\frac{n^{2/3}m^{1/3}\log^{1/3}(1/\delta)}{\eps^{4/3}}\right)$, we get that:
\[
\E[Z|F_x,F_y]=\Omega\left(\sqrt{\frac{k}{m}}\cdot \frac{\eps^2k^{3/2}}{n\sqrt{m}}\right)= \Omega\left(\sqrt{\frac{k}{m}}\cdot\frac{\eps^2 n\sqrt{m\log(1/\delta)}}{\eps^2 n\sqrt{m}}\right)=\Omega\left( \sqrt{(k/m) \log(1/\delta)} \right).
\]
\begin{comment}
\new{Using Lemma \ref{lem:ENq}, we get that  $\E[Z|F_x,F_y]=\Omega\left( \sqrt{N_q \log(1/\delta)} \right)$. Furthermore, by the definition of $k$, we have that $\frac{k^2}{nm}=\Omega(\log(1/\delta))$. Finally, whenever we get to compute $Z$ in algorithm \ref{alg:Basic-indep}, it holds that $N_p \leq  20N_q+ C\log(1/\delta)$ for some constant $C$. By combining all the above, we get that $\E[Z|F_x,F_y]=\Omega(\sqrt{N\log(1/\delta)})$.}
\end{comment}
\end{itemize}
\item \textbf{Case 3:}  $\E[Z|F_x,F_y]=\Omega\left(\frac{k^{3/2}\eps^2}{\sqrt{nm}}\right)$.

\dnew{We note that this is larger than the expression in Case 2 unless $k>nm$. 
Thus, it suffices to show that $\frac{k^{3/2}\eps^2}{\sqrt{nm}} = \Omega(\sqrt{k\log(1/\delta)})$. 
However, this follows from the fact that $k=\Omega\left( \frac{\sqrt{nm \log(1/\delta)}}{\eps^2} \right).$}

\begin{comment}
Using the fact that $k>\frac{\sqrt{nm\log(1/\delta)}}{\eps^2}$, and following arguments in the proof of Lemma \ref{lem:close-exp-gap}, we get that:
\[
\E[Z|F_x,F_y]=\Omega(\sqrt{k\log(1/\delta)})= \new{\Omega(\sqrt{N\log(1/\delta)})}
\]
\new{since $N\leq 100k$.}
%\[
%\E[Z|F_x,F_y]=\Omega(\sqrt{k\log(1/\delta)})=\Omega\left( \sqrt{(k/m) \log(1/\delta)}\right)
%\]
\end{comment}
\end{itemize}

Combining these two bounds, we get the required statement for any possible choice of flattening samples. 
Thus, the unconditional version of the statement also holds.
\end{proof}
This completes the proof of the lemma.
\end{proof}

Recall that the full algorithm is the following:
\begin{enumerate}
\item Let $\overline{S}$ be a random set of $100k$ samples.
\item Run \textsc{BasicTest} on $\overline{S}$ until it does not return \textrm{ABORT}.
\item Return ``YES"/``NO" as appropriate.
\end{enumerate}

\begin{lemma}\label{lem:ub-final}
If $p=q$ the probability that \textsc{FullTest} returns ``NO" is $O(\delta)$, 
and if $\dtv(p,q)>\eps$ the probability that it returns ``YES" is $O(\delta).$
\end{lemma}
\begin{proof}
\dnew{We bound the probability as follows:
\begin{align*}
\pr[\textrm{\textsc{FullTest} incorrect}] & = \sum_{t=0}^\infty \pr[\textrm{\textsc{BasicTest} Returns \textrm{ABORT} }t\textrm{ times and then returns wrong output}]\\
& = \sum_{t=0}^\infty \E_{\overline{S}}[\pr[\textrm{\textsc{BasicTest} returns \textrm{ABORT}}|\overline{S}]^t \pr[\textrm{\textsc{BasicTest} returns wrong output}|\overline{S}]]\\
& \leq  \sum_{t=0}^\infty 2^{-t} \E_{\overline{S}}[\pr[\textrm{\textsc{BasicTest} returns wrong output}|\overline{S}]]\\
& = \sum_{t=0}^\infty 2^{-t} \pr[\textrm{\textsc{BasicTest} returns wrong output}]\\
& = 2\pr[\textrm{\textsc{BasicTest} returns wrong output}] = O(\delta).
\end{align*}
}
\end{proof}

%Replacing $\delta$ by a smaller multiple of itself gives our final result.
\begin{proof}[Proof of Theorem \ref{thm:ub-main}]
By Lemma \ref{lem:ub-final}, we get that there exists some constant $c>0$, such that Algorithm~\ref{alg:Full-indep} %, which uses  $\Theta\left( n^{2/3} m^{1/3} \log^{1/3}(1/\delta) / \eps^{\new{4/3}} + ((nm)^{1/2}\log^{1/2}(1/\delta)+\log(1/\delta))/\eps^2 \right) \;.$ samples,
outputs ``NO" with probability at most $\delta^\prime=c\cdot\delta$ if $p$ is a product distribution, 
and outputs ``YES" with probability at most $\delta^\prime$ if $\dtv(p,p_x\times p_y)\geq \eps$. 
Since $\delta=\delta^\prime/c$, the sample complexity is:
\[\Theta\left( n^{2/3} m^{1/3} \log^{1/3}(1/\delta^\prime) / \eps^{\new{4/3}} + ((nm)^{1/2}\log^{1/2}(1/\delta^\prime)+\log(1/\delta^\prime))/\eps^2 \right) \;, \]
as desired. %Note that if $p$ is $\eps$-far from any product distribution, then $\dtv(p,p_x\times p_y)\geq \eps$ clearly holds.
\end{proof}

\section{Sample Complexity Lower Bounds for Closeness and Independence} \label{sec:lb-ind}

Our lower bound proofs follow a similar outline as those in
\cite{DK16}. The main point of the argument there was that we reduced
to the following problem: We have two explicit distributions on pseudo-distributions.
The former outputs independent pseudo-distributions; the latter outputs usually
far from independent pseudo-distributions. We pick a random
pseudo-distribution from one of these families, take $\Poi(k)$ samples
from it, hand them to the algorithm and ask the algorithm to determine
which ensemble we started with. It is shown that this is impossible to
do reliably by bounding the mutual information between the samples and
the bit determining which ensemble was sampled from.

This argument unfortunately, does not help much for high probability bounds.
Although, bounding the mutual information away from $1$ might in theory
show that you cannot distinguish with high probability, in practice,
our bounds will provide bounds on mutual information that is more than
$1$ given enough samples. To overcome this difficulty, we replace our bounds on
mutual information with bounds on KL-divergence. Unlike the mutual
information---which is bounded by 1-bit---the KL-divergence between our
distributions can become arbitrarily large. It is also not hard to see
that if two distributions can be distinguished with probability
$1-\delta$, the KL-divergence is $\Omega(\log(1/\delta))$.
(See Fact~\ref{lem:better-dtv-ub}.)

Given this modification, our ensembles are identical to the ones used in
\cite{DK16}. Furthermore, the bounding techniques themselves are very
similar, using essentially the same expression as an upper bound on
KL-divergence as was used as an upper bound on mutual information. One
other slight technical change is that we need to show that the
reduction to our hard problem still works for high probability
testing, which is not difficult but still needs a careful argument.

\subsection{Roadmap of this Section}

The structure of this section is as follows:
In Section~\ref{sec:lb-setup}, we explain in more detail the setup of the analysis
including the use of Poissonization and pseudo-distributions,
and show that proving lower bounds in such a setting implies lower bounds in the setting we care about.
In Section~\ref{sec:lb-assumps}, we prove that we can make certain assumptions about
the values of the relevant parameters without loss of generality
when showing certain parts of the lower bound.
In Section~\ref{sec:lb-instance}, we formally define the lower bound instances.
In Section~\ref{sec:lb-suitability}, we show that the lower bound instances
satisfy some basic necessary properties needed for them to actually imply the required
sample complexity lower bounds.

With this groundwork established, the remaining subsections are devoted
to proving the lower bounds on each of the terms in the lower bound expression.

\subsection{Analysis Setup}\label{sec:lb-setup}

We wish to show that testers with failure probability $\delta$ require at least $k$ samples
for the bound on $k$ given in Theorem~\ref{thm:lb-main-pd}.

Our goal is to exhibit a pair of ``distributions which output two-dimensional discrete distributions
with probability $\delta/100$'' $\mathcal{D}_1,\mathcal{D}_2$, where $\mathcal{D}_1$ is a product distribution
with probability of failure $\leq \delta/100$ and $\mathcal{D}_2$ is $\epsilon$-far from a product distribution
with probability of failure $\leq \delta/100$. More specifically, the dimensions of the distributions
are both $n \times m$, for some $n, m \in \Z_+$.

Finally, suppose that the total variation distance between $k_1$ samples generated from $\mathcal{D}_1$
versus $k_2$ samples generate from $\mathcal{D}_2$ is at most $\delta/100$.
(Here, $k_1,k_2$ are random variables which, with probability $\delta/100$, are within a constant factor of some parameter $k$.)
Then no tester can test independence with probability of failure $\delta$ with fewer 
than $\Omega(k)$ samples.

To show this, we only define $\mathcal{D}_1,\mathcal{D}_2$ implicitly via
``distributions which output two-dimensional pseudo-distributions with high probability.''
A pseudo-distribution is like a distribution, except the probabilities only need to sum to
within a constant factor of $1$ instead of $1$ exactly.

\begin{definition}\label{def:pd}
A discrete pseudo-distribution is a vector whose entries are nonnegative and sum to within a factor of $100$ of $1$.
%, for some global fixed constant.
\end{definition}

We define Poissonized sampling from a pseudo-distribution as follows.

\begin{definition}\label{def:pois-pdf}
For a pseudo-distribution $q$, we define \emph{$k$ Poissonized samples from $q$}
as the vector of random variables, where the $i$th entry is $\Poi(k q_i)$.
If $q$ is a two-dimensional pseudo-distribution, then we can analogously define this
as the two-dimensonial vector, where entry $(i, j)$ is $\Poi(k q_{ij})$.
\end{definition}

We will ultimately prove a lower bound on the natural generalization of independence testing to pseudo-distributions.
In other words, we will show that if one is given a pseudo-distribution on two variables (i.e., with domain $[n] \times [m]$),
which is either a product pseudo-distribution or $\epsilon$-far from a product pseudo-distribution,
then distinguishing these two cases takes a large number of Poissonized samples from the pseudo-distribution.
Here, a product pseudo-distribution is defined as follows and naturally generalizes the definition of a product distributions.

\begin{definition}\label{def:product-pd}
A product pseudo-distribution is a product measure that is a pseudo-distribution.
\end{definition}

%We let $\widetilde{\mathcal{D}}_1$ or $\widetilde{\mathcal{D}}_2$ denote the distributions on pseudo-distributions.
We will always use $p$ to denote a pseudo-distribution output by either of our distribution on pseudo-distributions,
or more generally, a measure which is a pseudo-distribution with high probability.
%Then these induce distributions on distributions $\mathcal{D}_1,\mathcal{D}_2$.

We note that proving a lower bound on independence testing for Poissonized pseudo-distributions
also implies the same lower bound---up to constant factors---on the non-Poissonized independence
testing problem for actual distributions. Formally, we have the following simple lemma:

\begin{lemma}\label{lem:pseudo-lb-to-lb}
Suppose there exists a tester for independence testing for distributions with failure probability $\delta<1/2$
using $k \geq 10^{10} \log(1/\delta)$ samples. Then there exists a tester for the Poissonized independence testing problem
for pseudo-distributions with failure probability $O(\delta)$, 
which uses $O(k)$ Poissonized samples from the pseudo-distribution.
Furthermore, the resulting tester still works with failure probability $O(\delta)$ 
for the even more general problem of testing independence, where the promises 
in the completeness and soundness cases only hold
with probability $1-O(\delta)$, instead of holding deterministically.
\end{lemma}

Our intended sample complexity lower bound in Theorem~\ref{thm:lb-main-pd} involves
an $\Omega(\log(1/\delta) / \epsilon^2)$ term. Hence, the parameter $k$ will satisfy the lower bound condition
in the above lemma statement. As such, we will work with Poissonized samples from pseudo-distributions
without loss of generality.

\begin{proof}[Proof of Lemma~\ref{lem:pseudo-lb-to-lb}]
Consider any hypothetical tester for the (non-Poissonized) independence testing problem for actual distributions.
It suffices to show that such a tester can also be used for the closely related problem of
testing independence for Poissonized pseudo-distributions,
with only constant factor losses in the relevant parameters, i.e., $n, m,\epsilon,\delta$.

The basic idea is to just run the hypothetical tester on the samples from the pseudo-distribution.
In order to get this working formally, some care is required.

The first issue is that the hypothetical tester is guaranteed to work on samples
from a distribution---whereas for the Poissonized problem involving pseudo-distributions,
we have Poissonized samples from a pseudo-distribution. Thus, we must first show that
we can either convert such samples of the latter form into the former form or interpret them as such.

There are two steps to doing this. The first step is to get rid of the Poissonization, which is done as follows:
If we are given $\Poi(k)$ samples $S$ from a pseudo-distribution, then by the definition of pseudo-distributions
and the concentration of Poisson random variables, we have that with probability of failure $\delta+e^{-k/10}$,
we will actually get at least $k/1000$ samples.
We can then run the hypothetical tester on this set $\bar{S}$ of $k/1000$ samples, discarding any extra samples.
By doing this, we are running the hypothetical tester on a fixed number of samples, eliminating the Poissonization.
Recall that the additional probability of failure resulting from this procedure is $\delta+e^{-k/10}$.
Since $k \geq 10^{10} \log(1/\delta)$ by assumption, the probability of failure is still $O(\delta)$.

The second step is to show that after we have eliminated Poissonization, we can interpret the set $\bar{S}$ of $k/1000$ samples
from the pseudo-distribution as samples from a corresponding actual distribution. To see this, recall that a pseudo-distribution is a measure. The only thing separating a pseudo-distribution from a probability distribution 
is that it might not be normalized (i.e., have sum of its point masses sum to one).
Thus, with every pseudo-distribution, we can naturally associate a corresponding probability distribution
obtained by normalizing the pseudo-distribution so that it does sum to one. One can then observe that the samples $\bar{S}$
do in fact come from this associated actual distribution. To see this, note that if we have a pseudo-distribution $p$
with associated actual distribution $\bar{p}$, i.e., $p = c \cdot \bar{p}$, then $k$ Poissonized samples from $p$ are
given by sampling element $i$ with multiplicity $\Poi(k \cdot p_i) = \Poi(k c \cdot \bar{p}_i)$.
Thus, Poissonized samples from a pseudo-distribution are equivalent to a different number of Poissonized samples
from the corresponding actual distribution. If we condition on the number of samples, then we get rid of the Poissonization
and just have samples from an actual distribution.

Thus, we can think of the hypothetical tester as being run on non-Poissonized samples from an actual distribution.
So, it suffices to prove that this actual distribution will be a product distribution (completeness case)
or far from a product distribution (soundness case) according to whether the pseudo-distribution it came from was.

One can easily verify that a pseudo-distribution is a product if and only if its corresponding actual distribution is also a product.
Thus, this tester is correct in the completeness case. For the soundness case, let $T$ denote any set of pseudo-distributions
which is closed under rescaling all probabilities by any amount that keeps the sum within a factor of $100$ of $1$.
One can check that if we have any pseudo-distribution $p$, then the $\ell_1$-distance of $p$
from the closest element in $T$---is up to constant factors---no more than a constant factor larger
than the $\ell_1$-distance between the actual distribution corresponding to $p$
and the set of actual distributions corresponding to $S$. In particular, this implies correctness in the soundness case
and completes the proof of Lemma~\ref{lem:pseudo-lb-to-lb}.
\end{proof}

By Lemma~\ref{lem:pseudo-lb-to-lb}, to establish our tight lower bound for independence testing, it suffices
to prove the following theorem:

\begin{theorem}[Sample Complexity Lower Bound for Independence Testing on Pseudo-distributions]\label{thm:lb-main-pd}
The task of $(\eps, \delta)$-independence testing given access to Poissonized samples from a pseudo-distribution
on $[n] \times [m]$, where $n \geq m$, requires at least $k$ Poissonized samples, where
\[
k \geq \Omega\left( \frac{n^{2/3} m^{1/3} \log^{1/3}(1/\delta)}{\epsilon^{4/3}} + \frac{\sqrt{n m \log(1/\delta) }}{\epsilon^2} + \frac{\log(1/\delta)}{\epsilon^2}\right).
\]
\end{theorem}

The proof of Theorem~\ref{thm:lb-main-pd} is fairly technical and is given in the following subsections.
The first term in the lower bound is obtained in Section~\ref{sec:first-term-lb}.
The second is obtained in Section~\ref{sec:second-term-lb}.

The final term is immediate from the fact that distinguishing a fair coin from an $\epsilon$-biased coin
can be reduced to independence testing with only constant factor loss in $\epsilon,\delta$
and any constant---or larger---domain size. (The reduction is simply to use the fact that independence testing
generalizes closeness testing which generalizes uniformity testing.)

Since closeness testing is reducible to independence testing over domain $[n] \times \{0, 1\}$
with only constant factor loss in the domain size, failure probability, and strength $\epsilon$ of the soundness guarantee,
our sample complexity lower bound for independence testing also proves a tight lower bound on the sample complexity of
closeness testing.

\begin{corollary}\label{cor:close-lb}
The task of $(\eps, \delta)$-closeness testing for a pair of distributions
with domain size $n$ requires the following number $k$ of samples:
\[
k \geq \Omega\left( \frac{n^{2/3}\log^{1/3}(1/\delta)}{\epsilon^{4/3}} + \frac{\sqrt{n \log(1/\delta) }}{\epsilon^2}
+ \frac{\log(1/\delta)}{\epsilon^2}\right) \;.
\]
\end{corollary}

\begin{proof}
Let $p, q$ be the distributions on $[n]$ and we want to test whether
$p=q$ (completeness) versus $\dtv(p,q) \geq \eps$ (soundness) with failure probability at most $\delta$.

We construct the following instance of independence testing on a distribution over  $[n] \times \{0, 1\}$.
We define $r = \frac{1}{2} \{0\} \times p + \frac{1}{2} \{1\} \times q$.
Then $p=q$ implies that $r$ is a product distribution, and $\dtv(p,q) \geq \epsilon$
implies that $r$ is $\Omega(\epsilon)$-far from any product distribution.
Furthermore, given sample access to $p, q$, one can simulate sample access to $r$
without any overhead in the number of samples.

Thus, if we can test whether $r$ is a product distribution versus $\Omega(\epsilon)$-far from
any product with failure probability $\delta$ using $k$ samples,
we can $(\eps, \delta)$-test closeness between $p$ and $q$ using $k$ samples.
\end{proof}

\subsection{Allowable Assumptions on the Parameters}\label{sec:lb-assumps}

When we prove each of the lower bound terms in Theorem~\ref{thm:lb-main-pd},
we can, without loss of generality, make certain assumptions on the relationship between $k,n,m,\delta,\epsilon$.
In particular, up to a factor of $3$, the sum of the lower bound terms is equal to their max.
Thus, when proving any particular term is a lower bound, it is equivalent to prove that the term
is the lower bound under the assumption that it is the max. Along these lines, we will prove a lower bound
\[
k \geq \Omega \left( \max\left( \frac{n^{2/3} m^{1/3} \log^{1/3}(1/\delta)}{\epsilon^{4/3}}, 10^6\frac{\sqrt{n m \log(1/\delta) }}{\epsilon^2}, 10^{18}\frac{\log(1/\delta)}{\epsilon^2}\right) \right) \;.
\]
Note also that when wish to prove a specific term is a lower bound, we only need to do so up to constant factors.

Finally, note that any lower bound that holds for some $n, m$ also holds with $n,m$ switched,
so we may assume $m \leq n$.

The above simple considerations have two immediate consequences in the assumptions we can make without loss of generality when we prove our lower bounds. For the $\sqrt{nm \log(1/\delta)} / \epsilon^2$ lower bound term, we may assume:
\begin{assumption}\label{assump:term-2}
For $k,\delta,n,m$, we have $k \leq nm / (2 \cdot 10^{6} \cdot \epsilon^2)$, $\log(1/\delta) \leq nm/10^{24}$, $m\leq n$, and $\epsilon^4 \leq 10^{36} \cdot (m/n) \cdot \log(1/\delta)$.
\end{assumption}
\begin{claim}\label{claim:assumps-2}
When proving the $\sqrt{nm \log(1/\delta)} / \epsilon^2$ lower bound,
we may assume without loss of generality that Assumption~\ref{assump:term-2} holds.
\end{claim}

A proof is provided in Appendix~\ref{sec:app-lb}.

For the $n^{2/3} m^{1/3} \log^{1/3}(1/\delta)$ term, we may assume:
\begin{assumption}\label{assump:term-3}
For $k,\delta,n,m$, we have $k \leq n /(2\cdot 10^{12})$, $\log(1/\delta) \leq \epsilon^4 n/(10^{36} m)$, and $m \leq \epsilon^4 n \leq n$.
\end{assumption}

\begin{claim}\label{claim:assumps-3}
When proving the $n^{2/3} m^{1/3} \log^{1/3}(1/\delta)$ lower bound,
we may assume without loss of generality that Assumption~\ref{assump:term-3} holds.
\end{claim}

A proof is provided in Appendix~\ref{sec:app-lb}. 
Also observe that this latter assumption is more restrictive than the earlier one.

\begin{observation}
Assumption~\ref{assump:term-3} implies Assumption~\ref{assump:term-2}.
\end{observation}

\subsection{Lower Bound Instance}\label{sec:lb-instance}

We construct two closely related lower bound instances, one for each of the first two terms.

We start by describing the first instance, used for the $n^{2/3}m^{1/3}\log^{1/3}(1/\delta) / \epsilon^{4/3}$
lower bound term. We construct the pseudo-distribution $p$ on $[n] \times [m]$ randomly as follows:
For each $i, j$, we select $p_{ij}$ i.i.d.
In both the completeness and soundness cases, with probability $k/n$, we set
$p_{ij} = 1/(km)$, where the decision is independent for distinct values of $i$,
but always the same decision for the same value of $i$, even for different values of $j$.
%Otherwise, flip a coin and set $p_{ij} = (1 \pm \epsilon) / (nm)$ according to the result. For this coin flip, it is done i.i.d. for each $i,j$ in the soundness case, whereas we only flip different coins for distinct $i$ in the completeness case.
Otherwise, we do the following: In the completeness case, we set $p_{ij} =  1/(nm)$.
In the soundness case, we flip a fair coin and set $p_{ij} = (1 \pm \epsilon) / (nm)$
according to the result.

We now describe the lower bound instance used for the $\sqrt{n m \log(1/\delta)} / \epsilon^2$ lower bound term.
Let $E_i$ denote the event that in the construction of the pseudo-distribution $p$,
we are in the second (probability $(1-k/n)$) case for index $i$.
Let $E$ denote the event that $E_i$ holds for all $i$.
The lower bound instance is simply to choose $p$ conditioned on $E$.

%\subsection{Relationship to Instance in \cite{DK16}}

%A similar lower bound instance was used in \cite{DK16}. Specifically, the instance in \cite{DK16} is the same except that the completeness case is to set $p_{ij}=1/(km)$ with probability $k/n$---which is the same as above---but otherwise sets $p_{ij}=1/(nm)$ in the complteness case instead of $(1\pm \epsilon)/(nm)$.

%The analysis will heavily rely on this close relationship.
%\begin{lemma}

%\end{lemma}

In either case, we take $k$ Poissonized samples from the distribution $p$.
We will prove a lower bound on the number of Poissonized samples
required for the independence testing problem; i.e., on the parameter $k$ such that $\Poi(k)$
samples cannot allow one to distinguish between the completeness and soundness cases
with probability of failure better than $\delta$.

\subsection{Suitability}\label{sec:lb-suitability}

We first need to show that $p$ is likely to be a pseudo-distribution,
so that it can be used with Lemma~\ref{lem:pseudo-lb-to-lb} to obtain lower bounds.

\begin{lemma}
The random measure $p$ is a pseudo-distribution with probability of failure $\leq \delta$
if the parameters satisfy Assumption~\ref{assump:term-3} and are sufficiently large.
\end{lemma}

\begin{proof}
Let $M$ denote the measure output by $p$. Note that by construction,
\[
1 \leq \E\left[\sum_{ij} M_{ij}\right] \leq nm \cdot \left[ \frac{1}{nm} + \frac{1}{nm} \right] \leq 2.
\]
Thus, we just need to show that $\sum_{ij} M_{ij}$ concentrates within a factor of $50$ of its expected value
with probability of failure at most $\delta$.

To do this, we start by noting that each term in $\sum_{ij} M_{ij}$ is always $\leq 2/(km)$ and has variance
\[
\var[M_{ij}] \leq \frac{k}{n} \cdot \frac{1}{k^2 m^2} + \frac{1}{n^2 m^2} = \frac{1}{k m^2 n} + \frac{1}{n^2 m^2} \leq \frac{2}{k m^2 n}.
\]
By Bernstein's inequality, we have with probability of failure $\delta$ that
\begin{align*}
\left| \sum_{ij} M_{ij} - \E\left[\sum_{ij} M_{ij} \right] \right| &\leq \sqrt{2 \var\left[\sum_{ij} M_{ij}\right] \cdot \log(1/\delta)}
+ \frac{2}{3} \cdot \frac{2}{km} \log(1/\delta) \\
&\leq 2 \sqrt{\frac{\log(1/\delta)}{k m^2 n}} + \frac{2}{3} \cdot \frac{2}{km} \log(1/\delta) \\
&\leq \frac{2}{k m^3} + \frac{4}{3m} \\
&= o(1).
\end{align*}
So we do concentrate within the desired range with the desired probability.
\end{proof}

\begin{lemma}
The random measure $p$ generated by conditioning on $E$ being true is a pseudo-distribution
with probability of failure $\leq \delta$ if the parameters satisfy Assumption~\ref{assump:term-2}.
\end{lemma}

\begin{proof}
Let $M$ denote the measure output by $p$. Note that by construction,
\[
1 \leq \E\left[\sum_{ij} M_{ij} \right] = nm \cdot \frac{1}{nm} = 1 \;.
\]
Thus, we just need to show that $\sum_{ij} M_{ij}$ concentrates within a factor of $100$
of its expected value with probability of failure at most $\delta$.

To do this, we start by noting that each term in $\sum_{ij} M_{ij}$ is always $\leq 2/(nm)$ and has variance
\[
\var[M_{ij}] \leq \frac{\epsilon^2}{n^2m^2} \;.
\]
By Bernstein's inequality, we have with probability of failure $\delta$ that
\begin{align*}
\left| \sum_{ij} M_{ij} - \E\left[\sum_{ij} M_{ij}\right] \right|
&\leq \sqrt{2 \var\left[\sum_{ij} M_{ij}\right] \cdot \log(1/\delta)}
+ \frac{2}{3} \cdot \frac{2}{nm} \log(1/\delta) \\
&\leq \sqrt{\frac{nm}{n^2m^2}} + \frac{2}{3} \cdot \frac{2nm}{100nm} \\
&\leq \frac{1}{25}
\end{align*}
for sufficiently large values of the parameters.
Therefore, we do concentrate within the desired range with the desired probability.
\end{proof}

\subsection{Lower Bound for First Term}\label{sec:first-term-lb}
We will use the following property of KL divergence.
\begin{fact}\label{fact:kl-products}
If we have product distributions $p_1 \times p_2$ vs. $q_1 \times q_2$, then
\[
D(p_1 \times p_2||q_1 \times q_2) = D(p_1 || q_1) + D(p_2 || q_2).
\]
\end{fact}

Let $X \in \Z^{n \times m}$ denote the random vector whose entry $X_{ij}$
is the Poisson random variable $\Poi(k \cdot p_{ij})$, where $p$ is generated
in the completeness case and we use the same $p$ for all entries of $X$.
Let $X_i$ denote the vector of all entries of $X$, where the first coordinate is $i$.
Let $Y$ be the same but for the soundness case.

Then it suffices to prove that
\begin{equation}
D(Y_i||X_i) \leq k^3 \epsilon^4 / (n^3 m) \;, \label{eq:desired}
\end{equation}
for all---or equivalently, any---$i$.

The reason this suffices is that then we would have by the two earlier properties
of KL-divergence (Facts~\ref{lem:better-dtv-ub} and~\ref{fact:kl-products}) that %\jnote{couldn't read Eric's comment about below equation}
\[
\log(2/\delta) \leq D(Y||X) = \sum_{i} D(Y_{i}||X_{i}) = n \cdot D(Y_i||X_i) \leq n \cdot k^3 \epsilon^4 / (n^3 m) \;.
\]
Solving for the Poissonized number of samples $k$, we get a lower bound of
\[
k \geq \Omega\left( \frac{n^{2/3} m^{1/3} \log^{1/3}(1/\delta)}{\epsilon^{4/3}} \right).
\]

%So, we just need to massage the proof of lemma 3.7 to show it gives an upper bound on KL.

To complete the lower bound proof, we need to prove Equation~\eqref{eq:desired}.

\paragraph{Bounding $D(Y_i||X_i)$ by a Rational Function.}

We can bound the KL divergence by a rational function as follows.

\begin{lemma}\label{lem:rational-bound}
For any pair of discrete random variables $A,B$, we have that
\begin{align*}
D(A||B), D(B||A) %&\leq D(X_i||Y_i) + D(Y_i||X_i) \\
&\leq D(A||B) + D(B||A) \\
&\leq \sum_v \frac{\pr[A=v]+\pr[B = v]}{4} \cdot \left[ \left( 1 - \frac{\pr[A=v]}{\pr[B=v]}  \right)^2 +
\left( 1 - \frac{\pr[B=v]}{\pr[A=v]} \right)^2 \right].
\end{align*}
Furthermore, this inequality holds term-wise between the sum defining KL-divergence and the RHS.
\end{lemma}
\begin{proof}
If we write the $KL$-divergence using its definition
as a sum, and then apply Claim~\ref{lem:kl-to-rational} term-wise
with $a = \pr[A=v]$ and $b=\pr[B=v]$, we obtain
\begin{equation}
2 \cdot [D(A||B)+D(B||A)] \leq \sum_v \pr[B = v] \cdot \left( 1 - \frac{\pr[A=v]}{\pr[B=v]}  \right)^2
+ \pr[A=v] \left( 1 - \frac{\pr[B=v]}{\pr[A=v]} \right)^2 \label{eq:tmp1} \;.
\end{equation}
Applying Claim~\ref{lem:relate-rational-ubs} to the RHS of Equation~\eqref{eq:tmp1} gives
\begin{equation}
2 \cdot [D(A||B)+D(B||A)] \leq \sum_v  \pr[A=v]  \cdot \left( 1 - \frac{\pr[A=v]}{\pr[B=v]}  \right)^2
+ \pr[B = v] \left( 1 - \frac{\pr[B=v]}{\pr[A=v]} \right)^2 \label{eq:tmp2} \;.
\end{equation}
Adding Equations~\eqref{eq:tmp1} and~\eqref{eq:tmp2} %and summing over all $v$
completes the proof of the lemma.
\end{proof}

\paragraph{Bounding the Rational Function by $k^3 \epsilon^4 / (n^3 m)$.}
In Lemma~\ref{lem:rational-bound}, we bounded from above $D(X_i||Y_i)$ by the rational function
\[
\sum_v \frac{\pr[X_i=v]+\pr[Y_i = v]}{4} \cdot \left[ \left( 1 - \frac{\pr[X_i=v]}{\pr[Y_i=v]}  \right)^2
+ \left( 1 - \frac{\pr[Y_i=v]}{\pr[X_i=v]} \right)^2 \right] \;.
\]
We will now show that the latter quantity is $\leq O(k^3 \epsilon^4 / (n^3 m))$.
In future equations, we drop the four in the denominator since we do not care about constant factors.

\begin{lemma}
Under Assumption~\ref{assump:term-3},
we have
\[
D(Y_i||X_i) \leq O(k^3 \epsilon^4 / (n^3 m)).
\]
\end{lemma}
\begin{proof}
In this proof, we use several facts from Appendix A.2 of \cite{DK16}, which are identified when they are used.

Consider the sum
\[
D(Y_i||X_i) = \sum_v \pr[Y_i=v] \cdot \log\left(\frac{\pr[Y_i=v]}{\pr[X_i=v]}\right) \;.
\]
We start by conditioning on whether $\|v\|_1 < 2$.
In the case this holds, by Lemma~\ref{lem:rational-bound},
we have that the terms with such $v$ contribute
\begin{align*}
&\sum_{\|v\|_1 < 2} \pr[Y_i=v] \cdot \log\left(\frac{\pr[Y_i=v]}{\pr[X_i=v]}\right) \\
&\leq \sum_{\|v\|_1 < 2} \frac{\pr[X_i=v]+\pr[Y_i = v]}{4} \cdot \left[ \left( 1 - \frac{\pr[X_i=v]}{\pr[Y_i=v]}  \right)^2 + \left( 1 - \frac{\pr[Y_i=v]}{\pr[X_i=v]} \right)^2 \right].
\end{align*}
It thus suffices to show that
\[
\sum_{\|v\|_1<2} \left( \pr[X_i=v]+\pr[Y_i = v] \right) \cdot \left[ \left( 1 - \frac{\pr[X_i=v]}{\pr[Y_i=v]}  \right)^2 + \left( 1 - \frac{\pr[Y_i=v]}{\pr[X_i=v]} \right)^2 \right] \leq O(k^3 \epsilon^4 / (n^3 m)) \;.
\]
We break the LHS above into two pieces and bound each piece from above
by $O(k^3 \epsilon^4 / (n^3 m))$. Specifically, we consider the following pieces
\begin{equation}
\sum_{\|v\|_1<2} \left[ \pr[X_i=v]+\pr[Y_i = v] \right] \cdot\left( 1 - \frac{\pr[X_i=v]}{\pr[Y_i=v]}  \right)^2 \label{eq:piece-1}
\end{equation}
and
\begin{equation}
\sum_{\|v\|_1<2} \left[ \pr[X_i=v]+\pr[Y_i = v] \right] \cdot\left( 1 - \frac{\pr[Y_i=v]}{\pr[X_i=v]}  \right)^2 \;. \label{eq:piece-2}
\end{equation}
The desired upper bound on Equation~\eqref{eq:piece-1} is explicitly proven in Appendix A.2 of \cite{DK16}, namely
\begin{equation}
\sum_{\|v\|_1 \leq 2} \left[ \pr[X_i=v]+\pr[Y_i = v] \right] \cdot\left( 1 - \frac{\pr[X_i=v]}{\pr[Y_i=v]}  \right)^2 \leq O(k^3 \epsilon^4 / (n^3 m)) \;.
\end{equation}
One can also observe that if one swaps the probabilities of the completeness and soundness cases in that argument,
it still goes through, and proves the same upper bound on Equation~\eqref{eq:piece-2}.

Thus, it remains to analyze the $\|v\|_1 \geq 2$ case. %We give the analysis for the $\|v\|_1 \geq 2$ case for \cref{eq:piece-1}. The proof of the case for \cref{eq:piece-2} is analogous. Then we wish to bound
Then we wish to prove that
\[
\sum_{\|v\|_1 \geq 2} \pr[Y_i=v] \cdot \log\left(\frac{\pr[Y_i=v]}{\pr[X_i=v]}\right) \leq O(k^3 \epsilon^4 / (n^3 m)).
\]
Just as above, we bound the sum via Lemma~\ref{lem:rational-bound} an break it into two parts. If we do so, we see that it suffices to bound
\begin{equation}
\sum_{\|v\|_1<2} \left[ \pr[X_i=v]+\pr[Y_i = v] \right] \cdot\left( 1 - \frac{\pr[X_i=v]}{\pr[Y_i=v]}  \right)^2 \label{eq:piece-1b}
\end{equation}
and
\begin{equation}
\sum_{\|v\|_1<2} \left[ \pr[X_i=v]+\pr[Y_i = v] \right] \cdot\left( 1 - \frac{\pr[Y_i=v]}{\pr[X_i=v]}  \right)^2 \label{eq:piece-2b}
\end{equation}
each by
\[
O(k^3 \epsilon^4 / (n^3 m)) \;.
\]
This bound on Expression~\eqref{eq:piece-1b} is explicitly proven in~\cite{DK16}.

To establish the same bound on Expression~\eqref{eq:piece-2b}, we show it is bounded by Expression~\eqref{eq:piece-1b} up to constant factors. To do so, note that the $v$th term of \eqref{eq:piece-2b} over that of \eqref{eq:piece-1b} simplifies to
\[
\pr[Y_i = v]^2 / \pr[X_i = v]^2. 
\]
With this, it suffices to show that the probability of any vector $v$ with $\|v\|_1 \geq 2$ 
is no more than a constant factor larger in the soundness case compared to the completeness case. 
To see this, note that
\begin{align*}
 &\pr[Y_i = v] \leq \pr[Y_i = v \text{ and } \bar{E}_i] + \leq \pr[Y_i = v \text{ and } E_i] \\
&\leq \pr[Y_i = v \text{ and } \bar{E}_i] \\
&+ \prod_{i=1}^m \frac{e^{-k(1+\epsilon)/nm} (k(1+\epsilon)/mn)^{v_i} + e^{-k(1-\epsilon)/nm}
(k(1-\eps)/mn)^{v_i}}{2 \cdot v_i!} \\
&\leq \frac{k}{ne} \prod_{i=1}^m
\frac{(1/m)^{v_i}}{v_i!} + \prod_{i=1}^m
\frac{(k(1+\epsilon)/nm)^{v_i}}{v_i!} \\
&\leq \frac{k}{ne} \prod_{i=1}^m
\frac{(1/m)^{v_i}}{v_i!} + \frac{2k}{n} \prod_{i=1}^m
\frac{(1/m)^{v_i}}{v_i!} \\
&\leq 2 \cdot \frac{k}{n}\prod_{i=1}^m \frac{e^{-1/m}(1/m)^{v_i}}{v_i!} & \text{since $k \leq n/2$ and $\|v\|_1 \geq 1$}  \\
&= 2e \cdot \pr[Y_i = v \text{ and } \bar{E}_i]) \\
&= 2e \cdot \pr[X_i = v \text{ and } \bar{E}_i] \\
&\leq O(\pr[X_i = v]).
\end{align*}
This completes the proof.
\end{proof}

% In fact, these conditonal distributions are precisely the same as the unconditional $X_i$ and $Y_i$ distributions used in the lower bound instance for the second term

%To prove the same bound holds for \cref{eq:piece-2}, we compare \cref{eq:piece-2} to \cref{eq:piece-1} termwise.

%Specifically, consider any term $v$. Then the ratio of any term from \cref{eq:piece-2} over the corresponding term from \cref{eq:piece-1} is
%\[
%\frac{\left[ \Pr[X_i=v]+\Pr[Y_i = v] \right] \cdot\left( 1 - \frac{\Pr[Y_i=v]}{\Pr[X_i=v]}  \right)^2}{\left[ \Pr[X_i=v]+\Pr[Y_i = v] \right] \cdot\left( 1 - \frac{\Pr[X_i=v]}{\Pr[Y_i=v]}  \right)^2} = \left(\frac{\Pr[Y_i=v]}{\Pr[X_i=v]}\right)^2.
%\]

%So as long as the PMF's of any particular entry of the vector of samples received in the completeness and soundness cases are within a constant factor, we are done. One can verify this is true by inspecting the construction of the lower bound instance. Specifically, we start by observing that in both the completeness and soundness cases, with probability $k/n$, the lower bound instance outputs a measure where each probability $1/(km)$ and $X_i,Y_i$ will have equal PMFs conditioned on this occurring. Otherwise, $X_i$ is a vector of $m$ i.i.d. Poissons $\Poi(k/(nm))$ and $Y_i$ is a vector of i.i.d. mixtures
%\[
%\frac{1}{2} \cdot \Poi\left(\frac{(1+\epsilon)k}{nm}\right) + \frac{1}{2} \cdot \Poi\left(\frac{(1-\epsilon)k}{nm}\right).
%\]

\subsection{Lower Bound for Second Term}\label{sec:second-term-lb}

Finally, in this subsection, we establish the second term of the lower bound.
Formally, we prove:

\begin{lemma}\label{lem:bounds-2}
Under Assumption~\ref{assump:term-2},
or the stronger Assumption~\ref{assump:term-3}, we have that
\[
D(Y_{ij}|E_i||X_{ij}|E_i) \leq O(k^2 \epsilon^4/(n^2m^2)) \;.
\]
\end{lemma}

\begin{proof}
By Claim~\ref{lem:kl-to-rational} and an analogous argument to that used to prove Lemma~\ref{lem:rational-bound},
we have that
\begin{align*}
&D(Y_{ij}|E_i||X_{ij}|E_i)\\
&\leq \sum_{a} \left(\pr[X_{ij}=\ell|E_i]+\pr[Y_{ij} = \ell|E_i]\right) \cdot \left[ \left( 1 - \frac{\pr[X_{ij}=\ell|E_i]}{\pr[Y_{ij}=\ell|E_i]}  \right)^2 + \left( 1 - \frac{\pr[Y_{ij}=\ell|E_i]}{\pr[X_{ij}=ell|E_i]} \right)^2 \right] \;.
\end{align*}
Once again, we break this up into two pieces, one for each of the squared terms in the second set of brackets.
Appendix A.1 of \cite{DK16} explicitly proves that the piece corresponding to the left term is
\[
O\left(\frac{k^2 \epsilon^4}{n^2 m^2}\right) \;
\]
when $k/mn \leq 1$. For $k/mn \geq 1$, let $\lambda = k/nm$. Note that in the domain of interest $\lambda \ll 1/\eps^2$. We note that
$$
\pr[X_{ij}=\ell|E_i] = \frac{e^{-\lambda}\lambda^\ell}{\ell!}
$$
and
$$
\pr[Y_{ij}=\ell|E_i] = \frac{e^{-\lambda}\lambda^\ell\left(e^{\eps\lambda}(1-\eps)^\ell+ e^{-\eps \lambda}(1+\eps)^\ell\right)/2}{\ell!}.
$$
We note that 
$$
\frac{\pr[Y_{ij}=\ell|E_i]}{\pr[X_{ij}=\ell|E_i]} = \frac{\left(e^{\eps \lambda}(1-\eps)^\ell+ e^{-\eps \lambda}(1+\eps)^\ell\right)}{2} = \frac{e^{a+b}+e^{a-b}}{2} \;,
$$
where $a = \ell \log(1-\eps^2)=O(\ell\eps^2)$ and $b = \lambda \eps+\ell\log\left( \sqrt{\frac{1-\eps}{1+\eps}}\right)=O(\eps(\lambda-\ell)+\ell\eps^3).$ We also note that
$$
|\log((e^{a+b}+e^{a-b})/2)| = |a+\log(\cosh(b))| = O(|a|+b^2).
$$
Applying this, we find that the log ratio of probabilities is
$$
O(\ell\eps^2 + (\lambda-\ell)^2 \eps^2 + \ell^2 \eps^6).
$$
We note that this is $O(1)$ whenever $|\lambda-\ell|=O(1/\eps)$. 
In this range, the contribution to our sum of KL-divergences is
$$
\sum_\ell \pr[X_{ij}=\ell|E_i]O\left(\log\left(\frac{\pr[Y_{ij}=\ell|E_i]}{\pr[X_{ij}=\ell|E_i]} \right) \right)^2.
$$
This is at most the expectation over $\ell \sim \mathrm{Poi}(\lambda)$ of $O(\ell^2 \eps^4 + (\lambda-\ell)^4\eps^4 + \ell^4\eps^{12})$, which is $O(\lambda^2\eps^4+\lambda^4\eps^{12}) = O(\lambda^2\eps^4) = O(k^2 \eps^4/(nm)^2)$, as desired.

Next we consider what happens when $|\ell-\lambda| \geq 1/\eps$. This is an event that happens with probability $\exp(-\Omega(1/(\eps \sqrt{\lambda}))) = O(\lambda \eps^2)^{10}$ under either $X$ and $Y$ conditioned on $E$. When this occurs, the contribution to the KL-divergence is $O(\ell\eps^2 + (\lambda-\ell)^2 \eps^2 + \ell^2 \eps^6).$ 
By Cauchy-Schwarz, the total contribution to the KL divergence from these terms is thus at most
$$
\sqrt{O(\lambda \eps^2)^{10} \E[O(\ell\eps^2 + (\lambda-\ell)^2 \eps^2 + \ell^2 \eps^6)^2]} = O(\lambda^2 \eps^4).
$$
This completes our proof.
\end{proof}

%\newpage

\bibliographystyle{alpha}
\bibliography{allrefs}

\appendix

\newpage

\section{Testing Properties of Collections} \label{sec:col}
\new{ In this section, we consider a generalization of the closeness testing problem, where the goal is to test whether a set of distributions are collectively close to a given distribution.}

\subsection{Problem Setup}

In this problem we are given $m$ different probability distributions $p^{(1)},\ldots,p^{(m)}$ over the domain $[n]$ and are guaranteed that either $p^{(1)}=p^{(2)}=\ldots = p^{(m)}$ or for any distribution $P$ that $\frac{1}{m}\sum_{i=1}^m \dtv(p^{(i)},P) \geq \eps$. The algorithm is given sample access to a pair $(i,x)$ where $i$ is drawn uniformly from $[m]$ and $x$ is drawn from $p^{(i)}$.

It can be easily seen that this is very similar to an independence testing problem. In particular, if we let $p$ be the distribution of pairs $(i,m)$ we are exactly being asked to distinguish between the cases where $p$ is independent and where $p$ is $\Omega(\eps)$-far from any independent distribution. %\textcolor{blue}{[TODO cite our independence testing paper for the lemma that says if $p$ is $\eps$-far from independent, it is $O(\eps)$-far from the product of its marginals]}.
\new{Note that, our tester outputs ``YES" if $p$ is independent, of ``NO" if $p$ is $\eps$-far from the product of its marginals, both with probability $1-\delta$. The case where $p$ is $\eps$-far from any product distribution, clearly falls into the latter case. } %\tnote{Which lemma is this? Isn't this obvious since the product of marginals is independent? I guess we need the converse of that.}
 There are two important differences between this problem and the standard independence testing problem. Firstly, we are guaranteed that the second marginal of $p$ is the uniform distribution. Secondly, we are no longer sorting the parameters so that $n\geq m$. Somewhat surprisingly, the final sample complexity is the same up to constants.

\new{
\begin{theorem}\label{thm:collections-main}
There exists a universal constant $C>0$ such that the following holds:
When
\begin{equation}
k \geq C \left( n^{2/3} m^{1/3} \log^{1/3}(1/\delta) / \eps^{\new{4/3}} + ((nm)^{1/2}\log^{1/2}(1/\delta)+\log(1/\delta))/\eps^2 \right) \;,
\end{equation}
Algorithm \textsc{BasicTestFamily} is an $(\eps, \delta)$-tester for testing closeness of collections of $m$ 
distributions on $[n]$ in total variation distance.
\end{theorem}
}

\subsection{Algorithm Outline}

The basic algorithm is the same as our algorithm for independence testing with one notable exception. Instead of obtaining samples from $q = p_x\times p_y = p_x \times U_m$ by taking two samples from $p$ and combining their $x$ and $y$ coordinates, instead we take just a single sample from $p$ and re-randomize its second coordinate. Because of this, it will no longer be necessary to flatten on the second coordinate. Our basic tester works as follows and the main tester is obtained from the basic tester in a way identical to the one used \new{in Section \ref{sec:indep-alg}.}

\begin{algorithm}\label{familyAlg}
    \SetKwInOut{Input}{Input}
    \SetKwInOut{Output}{Output}

    \Input{A Multiset $\overline{S}$ of $100k$ samples from $[n]\times [m]$ with $k=C\left(\frac{n^{2/3}m^{1/3}\log^{1/3}(1/\delta)}{\eps^{4/3}}+\frac{\sqrt{nm\log(1/\delta)}}{\eps^2}+\frac{\log(1/\delta)}{\eps^2}\right)$ where $C$ is a sufficiently large universal constant.}
    \Output{Information relating to whether these samples came from an independent distribution.}
    \tcc{Choose flattening $F$}
    $F_x\gets \emptyset$\\
  \For{$s \in \overline{S}$ }{
    $F_x=F_x\cup \{s\}$ with prob $\min\{n/100k,1/100\}$\\
    }

    \If{\label{flatteningfamily}$|F_x|>10n$}{\label{large Ffamily} return \textrm{ABORT}}

    \tcc{Draw samples $S_p^f, S_q^f$}
Let $\overline{S}^\prime=\{(x_i,y_i)\}$ be a uniformly random permutation of $\overline{S}\setminus (F_x\cup F_y)$\\
Draw $\ell,\ell^\prime\sim \Poi(2k)$.\label{permfamily}\\
\If{\label{samplingfamily}$\ell+\ell^\prime>|\overline{S}^\prime|$ }{\label{large lfamily} return \textrm{ABORT}}
Let $S_q=\{(x_j,u_j)\}_{j=1}^\ell$ where the $u_j$ are drawn i.i.d. uniformly from $[m]$. Let $S_p=\{(x_j,y_j)\}_{j=\ell+1}^{\ell+\ell^\prime}$\\
Create $S_p^f,S_q^f$ by assigning to corresponding sub-bins u.a.r\\
   Let $N_q:$ $\sharp$samples in $S_q^f$ that collide with another sample in $S_p^f\cup S_q^f$.\\
   Let $N_p:$ $\sharp$samples in $S_p^f$ that collide with another sample in $S_p^f$\\

\If{$N_q > c \max(k/m,k^2/mn)$}{\label{large Nqfamily} return \textrm{ABORT}}
  \If(\tcp*[f]{$C'$ a sufficiently large constant}){$N_p > 20N_q+ C'\log(1/\delta)$ \label{cond:large Nfamily}}
  {return ``NO"\label{line:reject0family}}
  \tcc{Compute test statistic $Z$}
   Flag each sample of $S_p^f,S_q^f$ independently with probability $1/2$.\\
   Let  $X^{(p0)}_i,X^{(q0)}_i$ be the counts for the number of times element $i$ appears \emph{flagged} in each set $S_p^f,S_q^f$ respectively and $X^{(p1)}_i,X^{(q1)}_i$ be the corresponding counts on \emph{unflagged} samples.\\
   \label{statisticfamily}Compute the statistic $Z=\sum_{i} Z_i$, where $Z_i=|X_i^{(p0)}-X_i^{(q0)}|+|X_i^{(p1)}-X_i^{(q1)}|-|X_i^{(p0)}-X_i^{(p1)}|-|X_i^{(q0)}-X_i^{(q1)}|$.  \\
   \eIf{\new{$Z< C^\prime\cdot \sqrt{\dnew{\min(k,(k^2/(mn)+k/m))}\log(1/\delta)}$}}{return ``YES"\label{line:acceptfamily}}{return ``NO"\label{line:reject1family}}
\caption{ \textsc{BasicTestFamily}($\overline{S}$): Given a joint distribution $p$ over $[n]\times [m]$ (with second marginal $U_m$) with marginals $p_x,p_y$, test if $p_x$ and $p_y$ are independent. }
\end{algorithm}

The analysis of this algorithm is very similar to that our of independence tester. In particular, we need to show two things about this basic tester:
\begin{enumerate}
\item For any set $\overline{S}$ the probability that the tester returns \textrm{ABORT} is at most $1/2$.
\item If $\overline{S}$ is a random set drawn i.i.d. from $p$ the probability that it returns a wrong answer (``NO" if $p$ is independent or ``YES" if $p$ is $\eps$-far) is at most $\delta$.
\end{enumerate}
Once we have this, the full tester will be the same as \new{Algorithm \ref{alg:Full-indep}} with analysis the same as \new{in Section \ref{sec:indep-alg}.}%[TODO reference lemma].
 
Most of the analysis also follows in the similar vein. The probability of incorrectly rejecting in line \ref{line:reject0family} is bounded by the same argument as used in \new{Lemma \ref{lem: large Np}}. The probability of returning a wrong answer in lines \ref{line:acceptfamily} or \ref{line:reject1family} are also the same (as we have the same bounds on $N$ as before).

Bounding the probability of an \textrm{ABORT} is mostly straightforward. The probabilities that $|F_x|$ or $\ell+\ell'$ are too big can be bounded easily by the Markov inequality much as before. The one remaining issue is to prove an analogue of \new{Lemma \ref{lem:ENq}} to bound the probability that we abort on line \ref{large Nqfamily}.

\subsection{Analysis of $N_q$}

Here we prove the following Lemma which along with Markov's inequality should be enough to complete our analysis:
\begin{lemma}
The $N_q$ computed in Algorithm \ref{familyAlg} satisfies
$$
\E[N_q|\overline{S}] = O(k/m+k^2/mn).
$$
\end{lemma}
\begin{proof}
For a sample $s\in \overline{S}$ let $N_s$ be $0$ unless $s$ is in either $S_p$ or contributing to an element of $S_q$ and otherwise let $N_s$ be the number of other samples of $S_q^f$ that collide with it. It is clear that
$$
N_q \leq \sum_{s\in \overline{S}} N_s.
$$
Therefore, it will suffice to prove that for each $s$ that $\E[N_s] = O(1/m+k/nm)$.

To do this, we let $C_X$ be the number of other elements of $\overline{S}$ with the same $x$-coordinate as $s$ and let \new{$F_X$} be the number of elements of $F_x$ with that value. We note that conditioning on $s$ not being used in flattening that $F_X$ is a binomial distribution $\Bin(C_X,\min(1/100,n/100k))$. It follows that $\E[1/(1+F_X)] = O(\max(1,k/n)/\new{C_X})$.

Conditioning on $F_x$, there are at most $C_X$ samples that could collide with $s$. Each has a probability of doing so that is at most $1/(m(1+F_X))$ (as if that sample contributes to an element of $S_q$ there is a $1/m$ chance that it is assigned the correct $y$-coordinate and a $1/(1+F_X)$ chance that its $x$-coordinate is assigned to the correct sub-bin). Thus, the expected number of samples that collide is at most
$$
\E[N_s] = \E_{F_x}[\E[N_s|F_x]] \leq \E_{F_x} [C_X/(m(1+F_X))] = O(C_X/m \max(1,k/n)/C_X) = O(\max(1/m,k/nm)),
$$
as desired.

This completes the proof \new{of Theorem \ref{thm:collections-main}}.
\end{proof}

\section{\new{Algorithm for Closeness Testing with Unequally Sized Sets of Samples}} \label{sec:close-unequal-alg}
\new{ In this section, we consider a different generalization of the closeness testing problem, where we are testing the closeness of two distributions, and we have access to a different number of samples from each of them. Having unlimited number of samples from one of the distributions, would be equivalent to identity testing, where one of the distributions is explicitly known. Thus, this can be viewed as an interpolation of the two problems.}
\subsection{Setup}

Here we are given sample access to two distributions $p$ and $q$ on $[n]$ and are guaranteed that either $p=q$ or $\dtv(p,q)>\eps$, and would like to distinguish between these possibilities. However, we are now given $O(K+k)$ samples from $q$ and $O(k)$ samples from $p$. This is particularly interesting when $K\gg k$ as we would like to know whether the extra samples can be used to reduce the sample complexity. In particular, we show that this can be distinguished so long as $k$ is a sufficiently large multiple of
$$
\frac{n \sqrt{\log(1/\delta)/\min(n,K)}}{\eps^2} + \frac{\log(1/\delta)}{\eps^2}.
$$

\new{
\begin{theorem}\label{thm:unequal-main}
There exists a universal constant $C>0$ such that the following holds:
When
\begin{equation}
k \geq C \left(n \sqrt{\log(1/\delta) / \min(n, K)} + \log(1/\delta)\right)/\eps^2 \;,
\end{equation}
Algorithm \textsc{BasicTestDifferentSamples} is an $(\eps, \delta)$--closeness tester in total variation distance, for discrete distributions over $[n]$
that draws $O(K+k)$ samples from one distribution and $O(k)$ samples from the other.
\end{theorem}
}

\subsection{Algorithm}

Our algorithm runs on the same basic principles as our algorithm for independence. We have a basic tester that takes a fixed set of samples and returns ``YES", ``NO" or \textrm{ABORT}. Our full tester takes a random set of samples and repeatedly runs the basic tester until it gets a non-\textrm{ABORT} outcome and returns that. We again need only guarantee that:
\begin{enumerate}
\item For any set $\overline{S}$ the probability that the tester returns \textrm{ABORT} is at most $1/2$.
\item If $\overline{S}$ is a random set drawn i.i.d. from $p$ the probability that it returns a wrong answer (``NO" if $p$ is independent or ``YES" if $p$ is $\eps$-far) is at most $\delta$.
\end{enumerate}

Our basic tester uses $O(k)$ samples from $p$ and $O(K+k)$ samples from $q$. It uses $\min(n,K)$ of these $q$-samples to flatten and then draws $\Poi(2k)$ samples from each of $p$ and $q$ to compute the statistics $Z,N,N_q$. If we run out of samples (or flatten using too many) we abort. We also abort if $N_q$ is too large. Then if $N_p \gg N_q$ we reject and otherwise reject or accept based on the size of $Z$. The full version of our basic tester is as follows:

\begin{algorithm}\label{diffsamplesAlg}
    \SetKwInOut{Input}{Input}
    \SetKwInOut{Output}{Output}

    \Input{A Multiset $\overline{S}=\overline{S_q}\cup \overline{S_p}$ of $100(k+K)$ samples from $q$ and $100k$ samples from $p$ with $K>k=C\left(\frac{n \sqrt{\log(1/\delta)/\min(n,K)}}{\eps^2} + \frac{\log(1/\delta)}{\eps^2}\right)$ where $C$ is a sufficiently large universal constant.}
    \Output{Information relating to whether $p=q$.}
    \tcc{Choose flattening $F$}
    $F\gets \emptyset$\\
  \For{$s \in \overline{S_q}$ }{
    $F=F \cup \{s\}$ with prob $\min\{n/(100|\overline{S_q}|),1/100\}$\\
    }

    \If{\label{flatteningdiffsamples}$|F|>n$ or $|F|>50K$}{\label{large Fdiffsamples} return \textrm{ABORT}}

    \tcc{Draw samples $S_p^f, S_q^f$}
    Let $\overline{S_q^\prime}=\overline{S_q}\setminus F$.\\
    Draw $\ell,\ell^\prime \sim \Poi(2k)$.\\
    \If{\label{samplingdiffsamples}$\ell>\min(|\overline{S_q^\prime}|,100k)$ or $\ell^\prime > |\overline{S_p}|$ }{\label{large ldiffsamples} return \textrm{ABORT}}
    Let $S_q$ be a set of $\ell$ random samples (taken without replacement) from $\overline{S_q^\prime}$ and $S_p$ a set of $\ell^\prime$ random samples from $\overline{S_p}$.
    Create $S_p^f,S_q^f$ by assigning to corresponding sub-bins u.a.r\\
   Let $N_q:$ $\sharp$samples in $S_q^f$ that collide with another sample in $S_p^f\cup S_q^f$.\\
   Let $N_p:$ $\sharp$samples in $S_p^f$ that collide with another sample in $S_p^f$\\

\If{$N_q > c \max(k^2/K,k^2/n)$}{\label{large Nqdiffsamples} return \textrm{ABORT}}
  \If(\tcp*[f]{$C'$ a sufficiently large constant}){$N_p > 20N_q+ C'\log(1/\delta)$ \label{cond:large Ndiffsamples}}
  {return ``NO"\label{line:reject0diffsamples}}
  \tcc{Compute test statistic $Z$}
   Flag each sample of $S_p^f,S_q^f$ independently with probability $1/2$.\\
   Let  $X^{(p0)}_i,X^{(q0)}_i$ be the counts for the number of times element $i$ appears \emph{flagged} in each set $S_p^f,S_q^f$ respectively and $X^{(p1)}_i,X^{(q1)}_i$ be the corresponding counts on \emph{unflagged} samples.\\
   \label{statisticdiffsamples}Compute the statistic $Z=\sum_{i} Z_i$, where $Z_i=|X_i^{(p0)}-X_i^{(q0)}|+|X_i^{(p1)}-X_i^{(q1)}|-|X_i^{(p0)}-X_i^{(p1)}|-|X_i^{(q0)}-X_i^{(q1)}|$.  \\
   \eIf{\new{$Z< C^\prime\cdot \sqrt{\dnew{(\min(k,(k^2/K+k^2/n))+\log(1/\delta))}\log(1/\delta)}$}}{return ``YES"\label{line:acceptdiffsamples}}{return ``NO"\label{line:reject1diffsamples}}
\caption{ \textsc{BasicTestDifferentSamples}($\overline{S}$): Given two distributions $p$ and $q$ on $[n]$ provide information as to whether $p=q$. }
\end{algorithm}

\subsection{Basic Algorithm Analysis}

Much of our analysis here is either easy or identical to the analysis of our other algorithms. First, we consider the probability of our algorithm aborting. We note that the expected size of $F$ is at most $n/100$, so the probability that $|F|>n$ is at most $1\%$ by Markov's inequality (bounding the probability of aborting on line \ref{large Fdiffsamples}). Similarly the expected size of $F$ is at most $K+k\leq 2K$ so there is at most a $4\%$ chance that it is bigger than $50K$. Additionally, since the expected size is at most $K+k$, with at least $98\%$ probability, we have that $|\overline{S_q^\prime}|\geq 50k$. Since, $\ell,\ell^\prime$ have expectations $2k$, if the above holds, there is at most a $6\%$ chance that either $\ell > 50k$ or $\ell' > 100k$, which bounds the total probability of aborting on line \ref{large ldiffsamples} by $8\%$. Bounding the probability of aborting on line \ref{large Nqdiffsamples} is more complicated and we will address it in the next section.

As for the probability of returning an incorrect result, the analysis from \new{Lemma \ref{lem: large Np}} still applies to show that we incorrectly reject on line \ref{line:reject0diffsamples} with probability at most $\delta$. The analysis for lines \ref{line:acceptdiffsamples} and \ref{line:reject1diffsamples} is more complicated and we will handle it shortly.

\subsection{The Expectation of $N_q$}

Appropriate bounds on probability of aborting on line \ref{large Nqdiffsamples} will follow from this lemma:
\begin{lemma}
For any given set $\overline{S} = (\overline{S_q},\overline{S_p})$, the expectation of $N_q$ is $O(\max(k^2/K,k^2/n))$, where $N_q$ is defined to be $0$ if the algorithm aborts before computing it.
\end{lemma}
\begin{proof}
Let $s$ be an element of $\overline{S}$ let $N_s$ be $0$ if $s$ is not chosen to be in $S_q$ or $S_p$ or if the algorithm aborts before computing $N_q$. Otherwise let $N_s$ be the number of other elements of $S_q^f$ that collide with it. It is not hard to see that $N_q \leq \sum_{s\in\overline{S}} N_s$. It is also not hard to see that the probability that $s$ is chosen in $S_p$ or $S_q$ is $O(1)$ and $O(k/(K+k))$, respectively. Therefore, our lemma will follow from the claim that for all $s$:
$$
\E[N_s | s\in S_p\cup S_q] = O(\max(k/K,k/n)).
$$

Let $C_X$ be the number of (other) elements of $\overline{S_q}$ with the same value as $s$ and let $F_X$ be the number of these samples in $F$. We note that $F_X$ is a binomial distribution with $C_X$ terms and probability $\Omega(\min(1,n/K))$ and so $\E[1/(1+F_X)] = O(\max(1,K/n)/C_X)$. Once we've conditioned on $F$, each of the remaining (at most $C_X$) elements of $\overline{S_q'}$ with the correct value have an $O(k/K)$ chance of being chosen to be in $S_q$ and then will have a $1/(1+F_X)$ chance of colliding with $s$ after flattening. Thus the total expected size is at most
\begin{align*}
\E[N_s | s\in S_p\cup S_q] = \E_{F}[\E[N_s | s\in S_p\cup S_q, F]] &\ll \E_F [C_X (k/K)/(1+F_X)]\\
 &\ll C_X (k/K) \max(1,K/n)/C_X\\
  &= \max(k/K,k/n),
\end{align*}
as desired.

This completes our proof.
\end{proof}

\subsection{Analysis of $Z$}

We now need to bound the probability of getting an incorrect output from lines \ref{line:acceptdiffsamples} or \ref{line:reject1diffsamples}. For this, we note an alternative scheme for generating $Z,N,N_q$ with $\overline{S}$ being taken at random. This can be done by taking $f\sim Bin(100K+100k, \min\{n/(100|\overline{S_q}|),1/100\})$, generating $f$ random samples from $q$ to get $F$. The letting $S_q$ and $S_p$ be generated by taking $\ell$ and $\ell'$ independent samples from $q$ and $p$, respectively. We note that this produces statistics $Z,N,N_q$ that are identically distributed to the ones produced by our algorithm when the algorithm doesn't abort before computing them. It will be enough for us to show that:
\begin{enumerate}
\item If $p=q$, the probability that $N_q < c\min(\max(k^2/K,k^2/n),k)$ and $N_p < 20 N_q + C' \log(1/\delta)$ and $Z> C'\sqrt{(\min(k,(k^2/K+k^2/n))+\log(1/\delta))\log(1/\delta)}$ is at most $\delta$.
\item If $\dtv(p,q)>\eps),$ and $|F|\leq 10n$ the probability that $N_q < c\min(\max(k^2/K,k^2/n),k)$ and $N_p < 20 N_q + C' \log(1/\delta)$ and $Z< C'\sqrt{(\min(k,(k^2/K+k^2/n))+\log(1/\delta))\log(1/\delta)}$ is at most $\delta$.
\end{enumerate}
(Note that if our algorithm doesn't abort that $N_q < |S_q| = O(k)$.)

Note that if the conditions $N_q < c\min(\max(k^2/K,k^2/n),k)$ and $N_p < 20 N_q + C' \log(1/\delta)$ hold that together they imply that $N=O(\min(\max(k^2/K,k^2/n),k)+\log(1/\delta))$. The first of these is easy to show, in particular, if $p=q$ we have that $\E[Z]=0$ and so by \new{Lemma \ref{lem:Zconc}, that} $Z \ll \sqrt{(N+\log(1/\delta))\log(1/\delta)}$ except with probability $\delta$.

The second condition is somewhat more involved. We will still have that $|Z-\E[Z|F]|\ll \sqrt{(N+\log(1/\delta))\log(1/\delta)}$ except with probability $\delta$. Furthermore, we know that since after conditioning on $F$, $Z$ is an instantiation of our closeness tester statistic on a domain of size $O(n)$ we have by \new{arguments in the proof of Lemma \ref{lem:close-exp-gap}} that
$$
\E[Z] = \Omega\left(\min\left(k\eps, \frac{k^2\eps^2}{n},\frac{k^{3/2}\eps^2}{n^{1/2}} \right)\right).
$$
It will thus suffice to show:
\begin{claim}
$$
\left(\min\left(k\eps, \frac{k^2\eps^2}{n},\frac{k^{3/2}\eps^2}{n^{1/2}} \right)\right) \gg \sqrt{(\min(k,(k^2/K+k^2/n))+\log(1/\delta))\log(1/\delta)}.
$$
\end{claim}
\begin{proof}
We split into cases based on whether $k>n$. In particular, if $k>n$, it suffices to show that
$$
\min\left(k\eps, \frac{k^{3/2}\eps^2}{n^{1/2}} \right) \gg \sqrt{k\log(1/\delta)}.
$$
However, we have that:
\begin{itemize}
\item Since $k\gg \log(1/\delta)/\eps^2$, we have $k\eps \gg \sqrt{k\log(1/\delta)}$.
\item Since $k\gg \sqrt{n\log(1/\delta)}/\eps^2$, we have $k^{3/2}\eps^2/n^{1/2} \gg \sqrt{k\log(1/\delta)}.$
\end{itemize}

For $k\leq n$, it will suffice to show that
$$
\min\left(k\eps, \frac{k^{2}\eps^2}{n} \right) \gg \sqrt{k^2/\min(K,n)\log(1/\delta)}+\log(1/\delta).
$$
However, we have that:
\begin{itemize}
\item Since $k\gg \log(1/\delta)/\eps^2$, we have $k\eps \gg \log(1/\delta)$.
\item Since $\min(n,K) > k \gg \log(1/\delta)/\eps^2$, we have $k\eps \gg \sqrt{k^2/\min(n,K) \log(1/\delta)}$.
\item Since $k \gg n\sqrt{\log(1/\delta)/\min(n,K)}/\eps^2$, we have $k^2\eps^2/n \gg \sqrt{k^2/\min(n,K)\log(1/\delta)}$.
\item Since $k \gg \sqrt{n\log(1/\delta)}/\eps^2$, we have $k^2\eps^2/n \gg \log(1/\delta)$.
\end{itemize}
This completes our proof \new{of Theorem \ref{thm:unequal-main}}.
\end{proof}

\section{Sample Complexity Lower Bound for Testing Closeness with Unequal Sized Samples}\label{sec:lb-unequal}

We now give a lower bound for closeness testing when one uses an unequal number of samples from each distribution. The overall technique is essentially the same as that from Section~\ref{sec:lb-ind}.

Suppose we are given discrete distributions $p,q$ which we want to test closeness with respect to which each have domain size $n$. We take $k$ Poissonized samples from $p$ and $K$ Poissonized samples from $q$. Without loss of generality, we may assume $K \geq k$.

For the hard family of instances, we construct a pseudo-distribution randomly by setting $p_i,q_i$ i.i.d for each $i$. Specifically, we set
\[
p_i = q_i = 1/K
\]
with probability $K/n$. Otherwise, we do the following. We set $p_i = \epsilon/n$ in both the completeness and soundness cases. We set $q_i = \epsilon / n$ in the completeness case and to either $0$ or $2 \epsilon/n$ in the soundness case with probability $1/2$ each.

%We can now prove the desired lower bound.

\begin{theorem}
Given discrete distributions $p,q$, which each have domain size $n$, let $K$ be taken from $p$ and $k$ from $q$. Then any closeness tester with failure probability better than $\delta$ requires the numbers of samples $k,K$ to satisfy
\[
k \geq \Omega\left(\frac{n\sqrt{\log(1/\delta)}}{\sqrt{K}\cdot\epsilon^2} + \frac{\sqrt{n\log(1/\delta)}}{\epsilon^2} + \frac{\log(1/\delta)}{\epsilon^2} \right).
\]
\end{theorem}

\begin{proof}
The last two terms are immediate from the fact that uniformity testing requires
\[
\Omega\left(\frac{\sqrt{n\log(1/\delta)}}{\epsilon^2} + \frac{\log(1/\delta)}{\epsilon^2}\right)
\]
samples \cite{DGPP17}. Thus, we only need to establish
\[
k \geq \Omega\left(\frac{n\sqrt{\log(1/\delta)}}{\sqrt{K}\cdot\epsilon^2} \right).
\]

We use the hard family of instances described earlier in this section and take $K$ Poissonized samples from $p$ and $k$ from $q$. Let $X_i$ be the $2$- tuple of the number of samples from element $i$ from each of $p$,$q$ in the completeness case and $Y_i$ be the same, but in the soundness case. Then we have by the product rule for KL divergence and Fact~\ref{lem:better-dtv-ub} that it suffices to show 
\[
D(X_i || Y_i) \leq O\left(\frac{\epsilon^4 k^2 K}{n^3}\right).
\]
By Lemma~\ref{lem:rational-bound}, we have
\begin{align*}
D(X_i || Y_i) &\leq \sum_v \frac{\pr[A=v]+\pr[B = v]}{4} \cdot \left[ \left( 1 - \frac{\pr[A=v]}{\pr[B=v]}  \right)^2 + 
\left( 1 - \frac{\pr[B=v]}{\pr[A=v]} \right)^2 \right].
\end{align*}
Thus, it suffices to show 
\[
\sum_v \frac{\pr[A=v]+\pr[B = v]}{4} \cdot \left[ \left( 1 - \frac{\pr[A=v]}{\pr[B=v]}  \right)^2 + 
\left( 1 - \frac{\pr[B=v]}{\pr[A=v]} \right)^2 \right] \leq O\left(\frac{\epsilon^4 k^2 K}{n^3}\right).
\]

To do this, we condition on the number of samples $v_2$ from $q$ according to whether $v_2=0$, $v_2=1$, or $v_2 \geq 2$ and bound the individual contributions to the sum from terms corresponding to each of these values of $v_2$. %For the $v_2=0$ case, we will also need to condition on whether $v_1=0$ or $v_1 \geq 0$.

However first, we establish some bounds that hold for all cases of this conditioning.
%For $v_2=0$,
First, note that since $n \geq K \geq k \geq \Omega( \sqrt{n} / \epsilon^2 )$ and $\epsilon < 1$, we have
\begin{align*}
\pr[X_i = v], \pr[Y_i = v] &= \frac{\left(1-\frac{K}{n}\right)(K \epsilon / n)^{v_1}(\Theta(k) \epsilon / n)^{v_2}e^{-(K+\Theta(k))\epsilon/n}+\frac{K}{n} (K/K)^{v_1}(k/K)^{v_2}e^{-(K+k)/K}}{v_1!v_2!}\\
&= \frac{\Theta(1) \cdot (K \epsilon / n)^{v_1}(\Theta(k) \epsilon / n)^{v_2}e^{-(K+\Theta(k))\epsilon/n}}{v_1!v_2!} + \Theta(1) \cdot \frac{K}{n \cdot v_1! \cdot v_2!}\cdot \left(\frac{k}{K}\right)^{v_2}\\
&= \Theta(1) \cdot \frac{(K \epsilon / n)^{v_1}(\Theta(k) \epsilon / n)^{v_2}}{v_1!v_2!} + \Theta(1) \cdot \frac{K}{n \cdot v_1! \cdot v_2!}\cdot \left(\frac{k}{K}\right)^{v_2}.
\end{align*}

Note that in the final line above, the first term dominates when $v=(0,0)$ and the second term dominates otherwise. Thus, regardless of what the $\Theta(1)$ terms in that line end up being, that line cannot change by more than a constant factor compared to any other choice of the $\Theta(1)$ terms. Thus, $\pr[X_i = v]$ and $\pr[Y_i = v]$ are within constant factors of each other.

Thus we may simplify the quantity we need to bound as
\begin{align*}
&\sum_v \frac{\pr[A=v]+\pr[B = v]}{4} \cdot \left[ \left( 1 - \frac{\pr[A=v]}{\pr[B=v]}  \right)^2 + 
\left( 1 - \frac{\pr[B=v]}{\pr[A=v]} \right)^2 \right] \\
&= \Theta(1) \cdot \sum_v \frac{(\pr[A=v]-\pr[B = v])^2}{\pr[A=v]+\pr[B = v]}.
\end{align*}

So it suffices to show
\[
\Theta(1) \cdot \sum_v \frac{(\pr[A=v]-\pr[B = v])^2}{\pr[A=v]+\pr[B = v]} \leq O(\epsilon^4 k^2 K / n^3).
\]

%We also have 
%\[
%\pr[X_i = v], \pr[Y_i = v] \geq O(1) \frac{\min\left[(K \epsilon / n)^{v_1},(1/K)^{v_2}\right]}{v_1!} + 
%\]

We also have
\begin{align*}
\left| \pr[X_i = v] - \pr[Y_i = v] \right| &\leq \frac{(\epsilon K / n)^{v_1} e^{-\epsilon K / n}}{v_1!} \cdot \left|\frac{(\epsilon k / n)^{v_2} e^{-\epsilon k / n}}{v_2!} - \frac{1}{2} \left[ \frac{(2\epsilon k / n)^{v_2} e^{-2\epsilon k / n}}{v_2!} + \delta_{v_2=0} \right] \right| \\
&\leq \frac{(\epsilon K / n)^{v_1}}{v_1!} \cdot \left|\frac{(\epsilon k / n)^{v_2} e^{-\epsilon k / n}}{v_2!} - \frac{1}{2} \left[ \frac{(2\epsilon k / n)^{v_2} e^{-2\epsilon k / n}}{v_2!} + \delta_{v_2=0} \right] \right|  %\\
%&\leq \left|\frac{(\epsilon k / n)^{v_2} e^{-\epsilon k / n}}{v_2!} - \frac{1}{2} \left[ \frac{(2\epsilon k / n)^{v_2} e^{-2\epsilon k / n}}{v_2!} + \delta_{v_2=0} \right] \right|
\end{align*}
where we used the fact that the contribution to the probabilities from the $q_i=p_i=1/K$ case cancels and in the other case, $(X_i)_1,(Y_i)_1$ have the same distribution which is independent of the other coordinate.

Thus, for $v_2=0$, we have
\begin{align*}
&\sum_{v|v_2=0} \frac{(\pr[X_i=v]-\pr[Y_i = v])^2}{\pr[X_i=v]+\pr[Y_i = v]} \\
&\leq O(1) \cdot \sum_{v|v_2=0} \frac{\frac{(\epsilon K / n)^{v_1}}{v_1!} \cdot \left|\frac{(\epsilon k / n)^{v_2} e^{-\epsilon k / n}}{v_2!} - \frac{1}{2} \left[ \frac{(2\epsilon k / n)^{v_2} e^{-2\epsilon k / n}}{v_2!} + \delta_{v_2=0} \right] \right|^2}{
\frac{(K \epsilon / n)^{v_1}(\Theta(k) \epsilon / n)^{v_2}}{v_1!v_2!}} \\
&\leq O(1) \cdot \sum_{v|v_2=0} O\left(\frac{\frac{(K \epsilon/n)^{v_1}}{v_1!} \cdot \left( e^{-k \epsilon/n} - (1+e^{-2k \epsilon/n})/2\right)^2}{
(K \epsilon/n)^{v_1} / v_1!}\right) \\
&\leq O(1) \cdot n (k \epsilon/n)^4 \\
&\leq O(1) & \text{since $k \leq n^{3/4}$}.
\end{align*}

For $v_2=1$, we have
\begin{align*}
&\sum_{v|v_2=1} \frac{(\pr[X_i=v]-\pr[Y_i = v])^2}{\pr[X_i=v]+\pr[Y_i = v]} \\
&\leq O(1) \cdot \sum_{v|v_2=1} \frac{\frac{(\epsilon K / n)^{v_1}}{v_1!} \cdot \left|\frac{(\epsilon k / n)^{v_2} e^{-\epsilon k / n}}{v_2!} - \frac{1}{2} \left[ \frac{(2\epsilon k / n)^{v_2} e^{-2\epsilon k / n}}{v_2!} + \delta_{v_2=1} \right] \right|^2}{\frac{K}{n \cdot v_1! \cdot v_2!}\cdot \left(\frac{k}{K}\right)^{v_2}} \\
&\leq O(1) \cdot \sum_{v|v_2=1} \frac{\frac{(\epsilon K / n)^{v_1}}{v_1!} \cdot ((\eps k/n)( e^{-k \eps/n} -e^{-2k \eps/n}))^2}{\frac{K}{n \cdot v_1!}\cdot \left(\frac{k}{K}\right)} \\
&\leq O(1) \cdot \sum_{v|v_2=1} \frac{(\epsilon K / n)^{v_1} \cdot ((\eps k/n)( e^{-k \eps/n} -e^{-2k \eps/n}))^2}{\frac{k}{n}} \\
&\leq O(1) \cdot \sum_{v|v_2=1} (\epsilon K / n)^{v_1} (\eps k /n)^4 (n/k) \\
&\leq O(k^3 \eps^4/n^3) \\ 
&\leq O(1).
\end{align*}

For $v_2 \geq 2$, we have
\begin{align*}
&\sum_{v|v_2\geq 2} \frac{(\pr[X_i=v]-\pr[Y_i = v])^2}{\pr[X_i=v]+\pr[Y_i = v]} \\
&\leq O(1) \cdot \sum_{v|v_2\geq 2} \frac{\frac{(\epsilon K / n)^{v_1}}{v_1!} \cdot \left|\frac{(\epsilon k / n)^{v_2} e^{-\epsilon k / n}}{v_2!} - \frac{1}{2} \left[ \frac{(2\epsilon k / n)^{v_2} e^{-2\epsilon k / n}}{v_2!} + \delta_{v_2=0} \right] \right|}{\frac{K}{n \cdot v_1! \cdot v_2!}\cdot \left(\frac{k}{K}\right)^{v_2}} \\
&\leq O(1) \cdot \sum_{v|v_2\geq 2} \frac{\frac{(\epsilon K / n)^{v_1}}{v_1!} \cdot ((\epsilon k / n)^{v_2} / v_2!)^2}{\frac{K}{n \cdot v_1! \cdot v_2!}\cdot \left(\frac{k}{K}\right)^{v_2}} \\
&\leq O(1) \cdot \sum_{v|v_2\geq 2} \frac{(\epsilon K / n)^{v_1} \cdot ((\epsilon k / n)^{v_2} / v_2!)^2}{\frac{K}{n \cdot v_2!}\cdot \left(\frac{k}{K}\right)^{v_2}} \\
&\leq O(1) \cdot \frac{n}{K} \sum_{v|v_2\geq 2} (\epsilon K / n)^{v_1} \cdot (\epsilon^2 kK / n^2)^{v_2} / v_2! \\
&\leq O(1) \cdot \frac{n}{K} \left. (\epsilon K / n)^{v_1} \cdot (\epsilon^2 kK / n^2)^{v_2} / v_2! \right|_{v=(0,2)} \\
&= O(1) \cdot \frac{n}{K} (\epsilon^2 kK / n^2)^{2} \\\
&= O(k^2 K \eps^4/n^3).
\end{align*}
This completes the case analysis and the proof.
\end{proof}

\section{Omitted Proofs from Section~\ref{sec:lb-ind}}\label{sec:app-lb}

\subsection{Allowable Assumptions}

\begin{proof}[Proof of \Cref{claim:assumps-2}]
For this lower bound term, we may assume
\[
k \leq \sqrt{nm \log(1/\delta)} / (2\epsilon^2)
\]
and
\[
100\log(1/\delta)/\epsilon^2 \leq 10\sqrt{nm \log(1/\delta)} / \epsilon^2 \;.
\]
Solving the second equation for $\log(1/\delta)$ yields
\[
\log(1/\delta) \leq nm / 100 \;.
\]
Substituting this into the first equation yields
\[
k \leq nm / (2\epsilon^2) \leq nm / 200 \;.
\]
\end{proof}

\begin{proof}[Proof of \Cref{claim:assumps-3}]
For this lower bound term, we may assume
\[
k \leq n^{2/3} m^{1/3} \log^{1/3}(1/\delta) / (2\epsilon^{4/3})
\]
and
\[
10\sqrt{nm \log(1/\delta)} / \epsilon^2 \leq n^{2/3} m^{1/3} \log^{1/3}(1/\delta) / \epsilon^{4/3}.
\]
Solving the second equation for $\log(1/\delta)$ yields
\[
\log(1/\delta) \leq \epsilon^4 m/n.
\]
Substituting this into the first equation yields
\[
k \leq m^{4/3}/(2n^{1/3}) \leq n/2.
\]
The final step follows from using $m \leq n$
\end{proof}

\subsection{Inequalities Used to Bound KL Divergence}

We first prove some simple inequalities.

\begin{claim}\label{lem:relate-rational-ubs}
For all $a,b > 0$, we have that
\[
\left[ b \left(1-\frac{a}{b}\right)^2 + a \left(1-\frac{b}{a} \right)^2 \right] \leq 
\left[ a \left(1-\frac{a}{b} \right)^2 + b \left(1-\frac{b}{a} \right)^2 \right] 
\]
\end{claim}
\begin{proof}
Moving everything to the same side and simplifying, the desired inequality is equivalent to
\[
\frac{(a - b)^2 (a + b)}{a^2 b^2} \geq 0.
\]
which is always true for $a,b>0$.
\end{proof}

\begin{claim}\label{lem:kl-to-rational}
For all $a,b > 0$, we have that
\[
a \log\left( \frac{a}{b} \right) + b \log\left( \frac{b}{a} \right) \leq  
\left[ b \left(1-\frac{a}{b} \right)^2 + a \left(1-\frac{b}{a} \right)^2 \right] /2 \;.
\]
\end{claim}

%One can show that no such inequality holds between any single term from the left and any single term on the right, even if we allow ourselves to lose arbitrary constant factors.

\begin{proof}
Note that the inequality we wish to prove is symmetric with respect to $a$ and $b$. 
Thus, we assume WLOG that $a \geq b$. Since both sides are $0$ when $a=b$, 
we assume WLOG that $a > b$.

The LHS is equal to
\[
(a-b) \log\left( \frac{a}{b} \right) \;.
\]
Dividing both sides by $(a-b)$, we get that we want to show
\[
\log\left( \frac{a}{b} \right) \leq \frac{b \left(1-\frac{a}{b} \right)^2 + a \left(1-\frac{b}{a} \right)^2}{2 (a-b)} = \left[ \frac{a}{b} - \frac{b}{a} \right] / 2.
\]
Define $x \triangleq a/b > 1$. 
Then we wish to show for all $x > 1$ that
\[
\log(x) \leq \left[ x - \frac{1}{x} \right]/2 \;.
\]
We prove this using standard calculus. 
Note that the two sides are equal when $x=1$. 
The derivative of the LHS is $1/x$ and the derivative of the RHS is $[1+1/x^2]/2$. 
It suffices to prove the former quantity is at most the latter for all $x$. So, we want to show
\[
\frac{1}{x} \leq \left[ 1 + \frac{1}{x^2} \right] / 2 \;, 
\quad \text{or equivalently,} \quad
1-2x+x^2 \geq 0 \;.
\]
The polynomial in the second equation is minimized by $x=1$ 
and takes value $0$ there, so it is nonnegative as desired.
\end{proof}

\end{document}